%% file: ndss_camera_ready.tex
\documentclass[conference]{IEEEtran}
\ifCLASSINFOpdf
  \usepackage[pdftex]{graphicx}
\else
\fi
\usepackage{amsmath}
\usepackage{amsthm}  % Add this alongside amsmath
\usepackage[inline]{enumitem}

\usepackage{booktabs}      % For \toprule, \midrule, \bottomrule
\usepackage{multirow}      % For \multirow command
\usepackage{array}         % Enhanced 
\usepackage{textcomp}      % For special text symbols table features
\usepackage{float}         % For better control of floating environments
\usepackage{diagbox}
\newtheorem{theorem}{Theorem}[section]
\newtheorem{lemma}[theorem]{Lemma}
\theoremstyle{definition}

\usepackage{newunicodechar}
\newunicodechar{ }{\,}

\usepackage{amssymb}
\usepackage{url}

\usepackage{algorithm}
\usepackage{algpseudocode}
\usepackage[caption=false,font=footnotesize]{subfig}

\usepackage{tikz}
\usetikzlibrary{positioning}
\begin{document}
\pagestyle{plain}
\title{Lightening the Load: A Cluster-Based Framework for A Lower-Overhead, Provable Website Fingerprinting Defense
}

\author{\IEEEauthorblockN{Khashayar Khajavi}
	\IEEEauthorblockA{Simon Fraser University\\
		kka151@sfu.ca}
	\and
	\IEEEauthorblockN{Tao Wang}
	\IEEEauthorblockA{Simon Fraser University\\
		taowang@sfu.ca}
	}

\IEEEoverridecommandlockouts
\makeatletter\def\@IEEEpubidpullup{6.5\baselineskip}\makeatother
\IEEEpubid{\parbox{\columnwidth}{
		Network and Distributed System Security (NDSS) Symposium 2026\\
		23-27 February 2026, San Diego, CA, USA\\
		ISBN 979-8-9919276-8-0\\
		https://dx.doi.org/10.14722/ndss.2026.241760\\
		www.ndss-symposium.org
}
\hspace{\columnsep}\makebox[\columnwidth]{}}

\maketitle

\begin{abstract}
Website fingerprinting (WF) attacks remain a significant threat to encrypted traffic, prompting the development of a wide range of defenses. Among these, two prominent classes are regularization-based defenses, which shape traffic using fixed padding rules, and supersequence-based approaches, which conceal traces among predefined patterns.
In this work, we present a unified framework for designing an adaptive WF defense that combines the effectiveness of regularization with the provable security of supersequence-style grouping.  
The scheme first extracts behavioural patterns from traces and clusters them into $(k,l)$-diverse anonymity sets; an early-time-series classifier (adapted from ECDIRE) then switches from a conservative global set of regularization parameters to the lighter, set-specific parameters.  
We instantiate the design as \emph{Adaptive Tamaraw}, a variant of Tamaraw that assigns padding parameters on a per-cluster basis while retaining its original information-theoretic guarantee. Comprehensive experiments on public real-world datasets confirm the benefits.  
By tuning $k$, operators can trade privacy for efficiency: in its high-privacy mode, Adaptive Tamaraw pushes the bound on any attacker's accuracy below \textbf{30\%}, whereas in efficiency-centred settings it cuts total overhead by \textbf{99} percentage points compared with classic Tamaraw.  
\end{abstract}


% \fi
\IEEEpeerreviewmaketitle

\input{Sections_revision/1_introduction}
\input{Sections_revision/2_related_Work}
\input{Sections_revision/3_threat_model}

\input{Sections_revision/4_problem_statement}
\input{Sections_revision/6_design_details}
\input{Sections_revision/7_experiments}

\input{Sections_revision/8_conclusion}

\bibliographystyle{IEEEtran}
\bibliography{references}

\input{Sections_revision/appendix}
\end{document}

%% file: Sections_revision/1_introduction.tex
\section{Introduction}
Tor is the leading low-latency anonymity network, relied upon by millions worldwide to shield their online activities from surveillance and censorship~\cite{dingledine2004tor}. Despite its robust encryption and onion routing design, Tor leaks metadata such as packet sizes, timing, and directional patterns that adversaries can exploit. Over the years, website fingerprinting (WF) attacks have demonstrated that by analyzing these residual traffic features, even passive attackers can infer with high accuracy which webpages users are visiting~\cite{wang2014effective,wang2016realistically,hayes2016k}. Modern deep learning techniques employing transformer models~\cite{vaswani2017attention} and multi-channel representations can now extract fine-grained patterns from traffic data, achieving higher recall and precision ~\cite{sirinam2018deep,rahman2019tik,shen2023subverting,mathews2024laserbeak}.

% In response to these threats, a variety of defenses have emerged to mitigate WF attacks. One broad class of approaches is regularization-based defenses, which aim to conceal traffic patterns by enforcing packet transmission rules to reduce entropy. One such rule is a fixed, constant packet rate,
% % regardless of the underlying webpage or traffic dynamics, 
% enforced by introducing artificial delays and dummy packets~\cite{cai2014systematic,cai2014cs,holland2020regulator}. While effective, these methods ignore the inherent variability in how different websites load and behave. As a result, they apply the same padding schedule to all traces, often leading to substantial and unnecessary overhead, particularly for short or bursty sessions.

In response to these threats, a variety of defenses have emerged to mitigate WF attacks. One broad class of approaches is regularization-based defenses, which aim to conceal traffic patterns by enforcing packet transmission rules to reduce entropy. A common rule is a fixed, constant packet rate,
% regardless of the underlying webpage or traffic dynamics, 
achieved with delays and dummy packets~\cite{cai2014systematic,cai2014cs,holland2020regulator}. These strategies often apply the same padding schedule to all sites, %disregarding the natural variability in website behavior. 
%This uniformity frequently 
which leads to excessive overhead, particularly for bursty traffic. A notable example is Tamaraw~\cite{cai2014systematic}, one of the few defenses to provide information-theoretic guarantees on attacker success. Tamaraw has not been broken; almost all defenses that lack such guarantees have been defeated by more powerful classifiers %that exploit subtle traffic cues beyond the reach of empirical tuning 
\cite{juarez2015wtf,gong2020zero,gong2022surakav}.

Another family of defenses aims to construct \textit{super-traces}, forcing instances of similar webpages to produce the same pattern and thus hinder an attacker from inferring the correct webpage~\cite{wang2014effective,nithyanand2014glove,wang2017walkie,shen2023subverting}.
However, these defenses are calibrated on a fixed set of sites.  
During an offline stage they partition that set into anonymity sets (clusters) with similar traffic characteristics and derive a canonical super‑trace for each cluster.  
At run time every page load is forced to follow the super-trace assigned to its destination page. 
A page that does not exist in the offline stage produces undefined behavior, but in practice, 
ordinary browsing covers far more destinations than any feasible reference set can capture.

In this work, we propose a hybrid website fingerprinting defense with the strengths of both regularization and supersequence approaches. Traffic loading under our defense begins with regularization: a global rate that protects early traffic without requiring prior knowledge of the destination. As the trace evolves, we switch to supersequence: a streaming early time series classifier assigns the live trace to a preconstructed anonymity set, after which the defense transitions to a lightweight regularization rate specific to that set. These sets are generated offline through clustering, subject to k-anonymity~\cite{sweeney2002k} and l-diversity~\cite{machanavajjhala2007diversity}.
%ensuring that all traces within a set are similar in structure while originating from diverse websites. 
Our construction satisfies an information theoretic upper bound on any adversary’s success. %through cluster specific adaptation.

Using this design framework, we create \textit{Adaptive Tamaraw}, a novel extension of the original Tamaraw defense that retains its provable security while reducing overhead. We provide a rigorous analysis of Adaptive Tamaraw that bounds any attacking classifier's maximum accuracy by the size and diversity of the anonymity set. We also evaluate our defense empirically: we test state-of-the-art website fingerprinting attacks to verify that the actual attack accuracy remains within our theoretical bounds while yielding lower overhead compared to the fixed-rate Tamaraw. %Further, Tamaraw suffers significantly if its fixed rates are misconfigured; Adaptive Tamaraw avoids this issue by basing its rate on anonymity sets of real traces.

% %Moreover, the underlying framework is general and can be applied to design other provable defenses against website fingerprinting.
% Our defense and its clustering phase are fully tunable, allowing users to adjust key parameters to strike the desired balance between overhead and security according to their specific operational needs.

We summarize the main contributions of our work as follows:
\begin{itemize}
    \item We propose a general design framework for constructing provable website fingerprinting defenses that combines regularization defenses with dynamic clustering to adaptively adjust defense parameters in real time. As a concrete instantiation, we develop \textit{Adaptive Tamaraw}, an extension of the original Tamaraw defense that preserves its information-theoretic guarantees while reducing bandwidth and latency overhead.
    
    \item We provide a formal analysis that quantifies the privacy guarantees of Adaptive Tamaraw. Specifically, we derive upper bounds on the maximum achievable attack accuracy, independent of the underlying classifier, based on the size and diversity of each anonymity set.
    
    \item Our experiments show that Adaptive Tamaraw offers flexible tunability, enabling practitioners to adjust the defense parameters to meet privacy and efficiency requirements. When configured for strong privacy, the defense reduces attack accuracy to below 30\%. In efficiency-focused settings, it achieves substantial reductions in the total overhead, up to 99 percentage points relative to the original Tamaraw. The implementation is publicly available at \url{https://github.com/khashayarkhaj/Adaptive-Tamaraw}.

\end{itemize}

%% file: Sections_revision/2_related_Work.tex
\section{Background and Related Work}

\subsection{Website Fingerprinting Attacks}
Website fingerprinting attacks enable an on-path adversary to infer which website a user visits by analyzing observable traffic features such as packet sizes, timings, and ordering, even when the content is encrypted. Early methods relied on manually engineered features combined with classical classifiers like SVMs~\cite{panchenko2016website}, k‑Nearest Neighbors~\cite{wang2014effective}, and k‑fingerprinting~\cite{hayes2016k}, the last of which uses random decision forests on manually extracted features.

The rise of deep learning has enabled WF attacks to automatically extract rich features from raw traffic. Deep Fingerprinting~\cite{sirinam2018deep} uses a CNN to capture local and global patterns, while Tik‑Tok~\cite{rahman2019tik} improves accuracy by incorporating timing and direction.  
Robust Fingerprinting (RF)~\cite{shen2023subverting} introduces the 2\,$\times$\,$N$ Traffic Aggregation Matrix, which counts inbound/outbound packets in 44\,ms slots and feeds it into a four-layer CNN, smoothing jitter while preserving bursts. RF surpasses 90\% closed-world accuracy and performs well in open-world settings, even against defenses like FRONT~\cite{gong2020zero}.  
LASERBEAK~\cite{mathews2024laserbeak} adds multi-channel features (timing, direction, burst edges) and attention layers to exceed 95\% accuracy. These attacks show that small architectural changes with strong features can defeat many prior padding defenses.

% The advent of deep learning has transformed the field by enabling automated extraction of rich representations directly from raw traffic data. Modern WF attacks build on these advances. Deep Fingerprinting~\cite{sirinam2018deep} employs a CNN architecture to capture both local and global traffic patterns, while Tik‑Tok~\cite{rahman2019tik} further improves accuracy by combining packet timing with directional information. Recent attacks raise the bar far beyond earlier CNN baselines.  
% Robust Fingerprinting (RF)~\cite{shen2023subverting} replaces raw
% packet sequences with a 2\,$\times$\,$N$ Traffic Aggregation Matrix that
% counts inbound and outbound packets in 44\,ms slots, then feeds this compact
% image into a four‑layer CNN.  The slotting smooths network jitter while
% preserving burst structure, letting RF exceed 90\,\% closed‑world accuracy and
% maintain high precision in a 40\,k‑site open‑world setting, even against
% defended traces such as FRONT.  
% LASERBEAK~\cite{mathews2024laserbeak} pushes further by stacking
% multi‑channel packet features (direction, timing, burst edges) with light‑weight
% attention modules, recovering over 95\,\% accuracy against the same defences.
% Together these models demonstrate that modest architectural tweaks, guided by
% careful feature design, can defeat most padding schemes previously considered
% state‑of‑the‑art.

\subsection{Website Fingerprinting Defenses}
When designing defenses against website fingerprinting attacks, the literature has broadly followed two lines of work. The first focuses on empirical defenses that are evaluated primarily through experiments and often rely on heuristics or traffic obfuscation strategies to reduce attack success. The second aims to provide formal, information-theoretic guarantees that rigorously bound the adversary’s ability to classify webpages. Table~\ref{tab:defence_overview} offers a comprehensive summary of both approaches, highlighting their key properties along with their strongest known attack accuracies. We now examine these two categories in more detail, beginning with empirical defenses.

\subsubsection{Empirical WF Defenses}

Many website fingerprinting defenses rely on empirical evaluation to show effectiveness against specific attacks. They often reduce fingerprinting accuracy and can be tuned for lower overhead but lack formal guarantees and remain vulnerable to stronger or adaptive adversaries. Empirical defenses fall into two main categories:
% Table~\ref{tab:defence_overview} provides a summary of representative empirical and formal defenses, along with their strongest known attack accuracies and security properties.

{\noindent \bf Obfuscation-Based Defenses}: These approaches inject randomness to obscure traffic patterns. WTF-PAD~\cite{juarez2015wtf} adds dummy packets during idle periods, FRONT~\cite{gong2020zero} perturbs early bursts with randomized padding, and Surakav~\cite{gong2022surakav} employs a GAN to simulate realistic timing. While effective against some attacks, they lack formal guarantees and are vulnerable to modern methods like Laserbeak~\cite{mathews2024laserbeak} and RF~\cite{shen2023subverting}. As shown in Table~\ref{tab:defence_overview}, these defenses still result in relatively high attack accuracies.

{\noindent \bf Regularization-Based Defenses}: These defenses reduce leakage by shaping traffic into fixed or structured patterns. BuFLO~\cite{dyer2012peek} and CS-BuFLO~\cite{cai2014cs} enforce constant rates, providing strong obfuscation but at high cost. RegulaTor~\cite{holland2020regulator} improves efficiency by reshaping only sensitive trace segments. However, most lack formal guarantees and thus can still leak expressive features. 

While empirical defenses can reduce current attack success, their lack of formal guarantees limits robustness against future threats. In contrast, as we will discuss,  Adaptive Tamaraw offers provable security~(Section~\ref{sec:max_accuracy}) and adaptively lowers overhead while preserving strong privacy guarantees.

\begin{table*}[t]
\centering
\caption{Strongest published attacks against representative WF defences in the closed-world setting.}
\vspace{-2pt}
\begin{itemize}
  \item Defences \textbf{without} a provable formal bound are consistently broken by adversarial training, with the attacker’s accuracy remaining well above the 50\,\% threshold commonly associated with robust resistance~\cite{shen2024real}.
  \item Super-sequence defences achieve low attack accuracy but protect only those websites present in their training corpus, limiting their applicability in open-world browsing.
\end{itemize}
\label{tab:defence_overview}
\begin{tabular}{@{}llllcc@{}}
\toprule
\textbf{Category} & \textbf{Defence} & \textbf{Strongest published attack\,$^\ddagger$} & \textbf{Accuracy (\%)} & \textbf{Formal bound?} & \textbf{Limited to Dataset?} \\ \midrule
\multirow{3}{*}{Obfuscation}
  & WTF-PAD          & RF~\cite{shen2023subverting}             & 96.6 & No  & No \\
  & FRONT            & LASERBEAK~\cite{mathews2024laserbeak}    & 95.9 & No  & No \\
  & Surakav          & LASERBEAK~\cite{mathews2024laserbeak}    & 81.5 & No  & No \\ \midrule
\multirow{3}{*}{Regularization}
  & RegulaTor        & RF~\cite{shen2023subverting}             & 67.4 & No  & No \\
  & Tamaraw          & LASERBEAK~\cite{mathews2024laserbeak}    & 25.3 & \textbf{Yes} & No \\
  & Walkie-Talkie    & RF~\cite{shen2023subverting}             & 93.9 & \textbf{Yes} & \textbf{Yes} \\ \midrule
\multirow{2}{*}{Super-sequence}
  & Super-Sequence   & Tik-Tok~\cite{rahman2019tik,shen2023subverting} & 29.18 & \textbf{Yes} & \textbf{Yes} \\
  & Palette          & RF~\cite{shen2023subverting,shen2024real}        & 36.43 & \textbf{Yes}$^{*}$ & \textbf{Yes} \\ \bottomrule
\end{tabular}
\vspace{3pt}

\footnotesize
$^{\dagger}$\;“Formal bound” indicates that the defence supplies an explicit, theoretical upper bound on an attacker’s success for any strategy.  
“Limited to Dataset” means the defence only applies to the websites contained in its training set, which is impractical for real-world browsing.\\
$^{*}$\;Palette does not provide a closed-form mathematical bound; its “Yes$^{*}$” entry denotes that its cluster-anonymisation method yields structured leakage reductions that approach provable guarantees~\cite{shen2024real}.\\
$^{\ddagger}$\;Where two references appear, the first paper introduces the attack and the second reports its performance against the listed defence.
\end{table*}

% Although these defenses can substantially lower an attacker’s success rate, they typically do not furnish formal bounds to any attackers success rate. Many rely on randomization, partial cover traffic, or approximate segmentatio, mechanisms that may still leak subtle features exploitable by advanced attacks. In contrast, Adaptive Tamaraw builds on a provably secure foundation, as discussed in Section \ref{sec:max_accuracy}, while introducing techniques to reduce overhead and maintain practical deployability.
\subsubsection{WF Defenses With Formal Bounds}

A number of defenses have sought not only to empirically reduce the success of website fingerprinting attacks but also to provide formal, provable bounds on security. Among these, Tamaraw~\cite{cai2014systematic} stands out as the first defense to derive an explicit upper bound on an adversary’s success probability. By construction, its uniform-rate padding guarantees that no attack, regardless of strategy, can exceed the bound as traces for different websites have the same rate and often have the same length. However, this theoretical rigor comes at a steep cost, and Tamaraw’s fixed-rate design introduces substantial overhead.

Other formal defenses attempt to balance rigor with efficiency. Walkie-Talkie~\cite{wang2017walkie}, for example, reduces traffic uniqueness by pairing sites and enforcing collisions in their trace representations. While it provides a bounded adversarial success rate in closed-world settings, it assumes prior knowledge of the entire trace; an alternative randomized version avoids this assumption, but loses its adversarial bound.

Wang et al. introduced the Supersequence defense~\cite{wang2014effective}, which groups webpages into anonymity sets and pads each trace to match a common super-sequence, ensuring that all traces within a cluster are indistinguishable. 
%This design yields a formal upper bound on attacker success, limited to $1/|C|$ where $|C|$ is the cluster size, assuming a uniform prior. The method also supports tunable trade-offs between overhead and security by adjusting clustering granularity and assumptions about client knowledge. 
Similarly, Palette~\cite{shen2024real} applies traffic cluster anonymization in real time by shaping traffic to match canonical burst patterns within each cluster, though it does not provide formal bounds.
%While it does not provide formal bounds, Palette's principled clustering and shaping reduce information leakage in a structured way, moving defenses closer to theoretical guarantees while maintaining practical deployability. 
However, these supersequence-based methods are limited to the websites present in their training dataset. 
%They cannot generalize to arbitrary page loads, which significantly limits their applicability in real-world scenarios.

% As summarized in Table~\ref{tab:defence_overview}, defenses offering formal bounds achieve low attack accuracies even after many years. Adaptive Tamaraw preserves provable guarantees while introducing clustering and adaptive switching mechanisms to reduce overhead. Crucially, unlike supersequence methods, it remains applicable to websites beyond the training set, making it suitable for real-world, out-of-training scenarios.

As summarized in Table~\ref{tab:defence_overview}, a clear dichotomy exists in the literature: while many low-overhead defenses have been proposed, those lacking formal guarantees have often been broken by subsequent attacks. In contrast, defenses with formal bounds, such as Tamaraw, have proven resilient over time but at the cost of high overhead. The core novelty of our paper lies in introducing a framework to bridge this gap between provable security and practical efficiency. In this work, we build upon Tamaraw specifically due to its unique and unbroken provably secure foundation, and demonstrate that its overhead can be reduced via dynamic clustering and switching mechanisms, all while preserving the original formal guarantees. Crucially, unlike supersequence methods, this approach remains applicable to websites beyond the training set, making it suitable for real-world, out-of-training scenarios.

%% file: Sections_revision/3_threat_model.tex
\section{Threat Model}
Following prior work~\cite{gong2022surakav, gong2020zero, shen2024real}, we assume a local, passive adversary who observes all traffic between the client and the Tor guard node. The adversary cannot modify traffic but records packet size, timing, and direction to infer the visited webpage. They can segment traffic into individual page loads, assuming users visit one page at a time. To execute the attack, the adversary collects training data by visiting monitored webpages and recording their traffic under similar conditions to the target user. The goal is to identify which monitored site the client is loading.

In this work, we apply two evaluation scenarios for the proposed defense: one where the user’s traffic is limited to webpages \emph{seen during the defense’s training phase} (\textit{in-training webpages}), and another where the user can also visit \emph{webpages not included in the defense training set} (\textit{out-of-training webpages}). In this latter case, we examine how the defense performs when users browse pages that were not part of the training data used to construct the anonymity sets or train the classifiers. This setting reflects real-world usage, where users may access arbitrary webpages beyond the training set, allowing us to assess the defense’s ability to protect previously unseen traffic patterns. 

We also assume the defense operates without prior knowledge of the destination webpage. 
Our design goal is to develop an easy-to-implement defense that is generally suitable for browsing any site, including dynamic webpages generated on the fly.
As such, we do not limit the defense's applicability to a pre-configured set of static pages and we also avoid the privacy risk of disclosing the destination to a defending proxy. 
As we will elaborate upon in Section \ref{sec:pattern_detection}, our approach is designed to operate on the more fundamental and recurring traffic patterns that emerge during a page load.

%% file: Sections_revision/4_problem_statement.tex
\section{Problem Statement}
\label{sec:problem_statement}

When a client loads a webpage, it produces a sequence of encrypted packets called a \emph{traffic trace}:
\[
f = \bigl[(t_1, d_1), (t_2, d_2), \ldots, (t_N, d_N)\bigr],
\]
where each packet \(f_i\) has a timestamp \(t_i\) and direction \(d_i \in \{+1, -1\}\) (outgoing/incoming). Since Tor pads packets to a fixed size, the adversary observes only timing and direction.

A defense modifies the trace \(f\) to produce a defended trace \(f'\) by delaying real packets or adding dummies, obscuring patterns used by WF attacks.  We evaluate the overhead of defenses by \textit{bandwidth overhead} and \textit{time overhead}. The bandwidth overhead is measured as the total number of dummy packets divided by the total number of real packets over the whole dataset. Similarly, the time overhead is measured as the total extra time divided by the total loading time in the undefended case over the whole dataset.
% We evaluate defenses using:
% \begin{itemize}
%     \item \textbf{Bandwidth Overhead:}
%     \[
%     B(f, f') = \frac{|f'| - N}{N},
%     \]
%     the ratio of dummy to real packets.
    
%     \item \textbf{Time Overhead:}
%     \[
%     T(f, f') = \frac{t' - t_N}{t_N},
%     \]
%     the relative delay of the last real packet, where  $t'$ is the timestamp of the last real packet in the defended trace.
% \end{itemize}

In this paper, we aim to address three limitations of existing WF defenses. 
First, static padding schedules that ignore real-world traffic diversity are inefficient. 
Second, supersequence-based defenses struggle in \emph{out-of-training} settings with unseen websites.  
Third, defenses that lack provable security guarantees are frequently broken by later attacks. %we formalize \emph{bounded security} using a weighted non-injectivity framework to bound attacker success.  
The following expands on each challenge.
% First, we formalize the notion of \emph{bounded security}, providing theoretical guarantees on the adversary’s success rate via a weighted non-injectivity framework.  
% Second, we examine the inefficiency of static padding schedules that fail to adapt to the heterogeneity of real-world traffic.  
% Third, we analyze the limitations of supersequence-based defenses in \emph{out-of-training} scenarios, where precomputed mappings cannot accommodate websites not included during the defense's training phase.  
% The following subsections elaborate on each of these challenges.

% Besides bounded security, our design aims to achieve two additional properties. % persist in the design of effective WF defenses. 

\subsubsection*{1. Excessive Overhead from Static Regularization}
Regularization defenses (such as Tamaraw) shape the traffic so that traces appear uniform, minimizing the features an attacker might exploit. 
However, these methods typically apply a fixed set of parameters statically across all traffic. This static configuration does not account for the inherent
variability in traffic patterns across different webpages, which can result in either
excessive overhead or insufficient obfuscation in certain
cases. %across different webpages and even within a single webpage, 
%or insufficient obfuscation in certain cases.
% On the other hand, RegulaTor flexibly adjusts its parameters during page loading, but in doing so sacrifices bounded security. 
% We demonstrate a design that maintains bounded security while {\bf avoiding the excessive overhead of a static page load rate}.

\subsubsection*{2. Poor Generalization in Supersequence-Based Defenses}
Supersequence-based approaches attempt to create anonymity sets by mapping each webpage to a common supertrace. While this strategy can provide strong theoretical guarantees, it is only effective on the set of \emph{in-training} webpages used during defense construction. If a user's traffic originates from an \emph{out-of-training} webpage, there is no reliable mechanism to map the trace to a corresponding supersequence. This restriction greatly limits the applicability of such defenses in real-world deployments. We demonstrate that our method generalizes effectively to \emph{out-of-training} webpages, making it more suitable for practical scenarios involving previously unseen inputs.

\subsubsection*{3. Lack of Provable Security Guarantees}

Existing defenses have usually been defeated by later-published attacks. In particular, as it can be seen in Table~\ref{tab:defence_overview}, state-of-the-art obfuscation defenses (e.g., FRONT, Surakav) often achieve relatively low overhead but have been circumvented by advanced WF attacks. A key limitation of these defenses is the lack of \emph{provable security}; they do not provide a theoretical upper bound on the adversary’s maximum success rate.

To provide an upper bound on attacker success, we start by analyzing the 
\emph{non-injectivity} of a defense.  
Let \(D\) map an input trace \(f\) to an output trace \(f' = D(f)\).  
For every output, the pre-image is  
\[
  D^{-1}(f') \;=\; \{ f : D(f)=f'\}.
\]

A defense is called \(\delta\)-non-injective with respect to a given defended trace \(f'\) if \(\bigl|D^{-1}(f')\bigr| \ge\delta\).  
An attacker who sees \(f'\) must then guess among at least \(\delta\) candidates to infer $f$.
%, so its success is capped by \(1/\delta\).
However, multiplicity alone is not sufficient for website fingerprinting security: the attacker’s success also depends on \emph{how those inputs are distributed across webpages}.  
If one webpage dominates the pre-image, the majority vote still yields a high success rate, even when \(\bigl|D^{-1}(f')\bigr|\) itself is large.  
The relevant quantity is label diversity: the ratio between the total pre-image size and the largest webpage-specific share.

To capture both multiplicity and label diversity, we refine the measure.
Let  
\(D^{-1}_{w}(f') = \{\,f\in D^{-1}(f') : \text{class}(f)=w\}\)
be the pre-image subset that originates from webpage \(w\).
We define the \emph{weighted pre-image size}
\begin{equation}
  \tilde{\delta}(f')
  \;=\;
  \frac{|D^{-1}(f')|}
       {\displaystyle\max_{w} |D^{-1}_{w}(f')| }.
\label{eq:weighted-delta}
\end{equation}
Intuitively, \(\tilde{\delta}(f')\) counts how many distinct inputs are
merged \emph{per majority webpage}.  
If every input in the bucket comes from the same site, the denominator equals the numerator and \(\tilde{\delta}(f') = 1\), so a perfect information attacker will always guess the class of $f'$ correctly. 

\smallskip
\textbf{Weighted \(\delta\)-non-injectivity.}  
We call a defense \emph{weighted \(\delta\)-non-injective} with respect to a given defended trace \(f'\) if
\(\tilde{\delta}(f') \ge \delta\).
Since an optimal attacker chooses the majority webpage, its success rate is
\(\max_{w}|D^{-1}_{w}(f')|/|D^{-1}(f')| = 1/\tilde{\delta}(f')\).
Thus, a weighted \(\delta\)-non-injective defense guarantees
\[
  \Pr[\text{attacker succeeds} \mid f'] \;\le\; \frac{1}{\delta}.
\]

\subsubsection*{Non-Uniformly Weighted $\delta$-Non-Injectivity}

While the above definition bounds the attacker's success on an individual trace, prior works and subsequent WF studies ~\cite{shen2023subverting, mathews2024laserbeak, cai2014systematic} focus on the \emph{average} success rate over the distribution of traces. To capture this, we define the \emph{non-uniformly weighted \(\delta\)-non-injectivity} property.

Let \(\mathcal{F}'\) denote the support of defended outputs, and let \(P(f')\) be the probability of observing a defended trace \(f'\) with \(f' \in \mathcal{F}'\).
We say the defense is \emph{non-uniformly weighted \(\delta\)-non-injective} if this inequality holds:
% Then the average attacker success rate is  
\[
\mathbb{E}_{f' \in \mathcal{F}'}\left[\Pr[\text{attacker succeeds} \mid f']\right] = \mathbb{E}_{f' \in \mathcal{F}'}\left[\frac{1}{\tilde{\delta}(f')}\right] \le \frac{1}{\delta}.
\]  
 This definition implies that the attacker’s success rate, on average, remains below the inverse of \(\delta\). This choice aligns with previous work on website fingerprinting defenses, which report attacker performance in terms of average success rate over the trace distribution, i.e., they evaluate \emph{non-uniform} security. Uniform security, while stricter, is rarely achieved in practice and not the focus of most deployed defenses. We therefore adopt the non-uniform formulation to remain consistent with established evaluation methodology~\cite{cai2014systematic}.

%% file: Sections_revision/6_design_details.tex
\section{Design Details}

\subsection{Overview}

In this work, we introduce the first website fingerprinting defense that performs \textit{adaptive parameter selection} while providing a formal upper bound on attacker success. Our design is motivated by three goals:
(i) replacing static regularization configurations with traffic adaptation,
(ii) generalizing to webpages not seen during training, and
(iii) ensuring provable security guarantees.

We do so using a {\bf global-to-local defense strategy}. Our defense assumes no knowledge of the page under load, so it always starts with a global strategy based on a regularization defense  (such as BuFLO~\cite{dyer2012peek}, Tamaraw~\cite{cai2014systematic}, or RegulaTor~\cite{holland2020regulator}).
We begin each page load using a set of global parameters that apply conservative padding and scheduling to protect early traffic without knowing the destination. After enough of the trace has been observed to confidently identify an anonymity set, we switch to a more efficient set of padding parameters tailored to that set.
This two stage strategy ensures strong privacy guarantees during the early ambiguous phase while reducing overhead in the remainder of the trace.

To concretely instantiate this global-to-local strategy, we build on the Tamaraw defense, leading to what we call \textit{Adaptive Tamaraw}.

\subsection{Adaptive Tamaraw}
\label{sec:adaptive-tamaraw}

To evaluate the general framework introduced in this paper, we use a concrete, well-studied regularization defense in our global-to-local strategy: Tamaraw~\cite{cai2014systematic}. Among existing regularization-based defenses, Tamaraw stands out as the only one that offers a formal, information-theoretic bound on an adversary's success rate. As shown in Table~\ref{tab:defence_overview}, it also yields the lowest attacker accuracy compared to other defenses. 
%This combination of provable security and empirical resilience makes Tamaraw an ideal foundation for demonstrating the benefits of our adaptive defense strategy.

Classic Tamaraw transmits fixed-size cells at two constant rates: $\rho_{\text{out}}$ seconds between successive upstream cells and $\rho_{\text{in}}$ seconds between successive downstream cells. Padding continues in each direction until its cell count reaches the next multiple of a hyperparameter $L$. Since $(\rho_{\text{out}}, \rho_{\text{in}})$ are fixed and identical across all traces, the resulting shapes of the defended traces are uniform in timing and rate. The only distinguishing feature that may remain between traces is their total length: specifically, which multiple of $L$ they are padded to.

Tamaraw requires a single pair of padding parameters $(\rho_{\text{out}}, \rho_{\text{in}})$ to be applied uniformly across all webpages. This global configuration leads to inefficiencies, as some webpages end up causing far more dummy traffic than necessary. As noted by~\cite{cai2014systematic}, the standard method for configuring Tamaraw involves fixing the $L$ and then performing a grid search over candidate $(\rho_{\text{out}}, \rho_{\text{in}})$ pairs. Each pair yields a specific tuple of bandwidth and time overhead, and the set of resulting points is filtered to retain only the Pareto-optimal ones (those that are not dominated in both overhead dimensions). %From this Pareto set, a global configuration $(\rho_{\text{in}}^g, \rho_{\text{out}}^g)$ is selected to reflect the desired trade-off between bandwidth and latency.

% Because a single tuple of padding parameters $(\rho_{\text{out}}, \rho_{\text{in}})$ must suit every site, many pages carry excessive dummy traffic. As mentioned in~\cite{cai2014systematic}, the standard approach to configure Tamaraw involves fixing the quantum $L$ and performing a grid search over candidate pairs $(\rho_{\text{out}}, \rho_{\text{in}})$ to balance the trade-off between time and bandwidth overheads. To be more precise, each pair produces a corresponding tuple of bandwidth and time overheads, from which we retain the Pareto-optimal set of overhead pairs. From this set, we select the global parameters $(p_{\text{in}}^{g}, p_{\text{out}}^{g})$ that match a desired bandwidth-time overhead trade-off.

In Adaptive Tamaraw, we start each trace with these {\em global} parameters, and then switch to {\em local} padding parameters when possible. This strategy is governed by three processes that will be described in the subsequent subsections. 

\begin{enumerate}

    \item We analyze the webpages in the training set of the defense to identify representative {\em traffic patterns}, which we call {\bf Intra-Webpage Pattern Detection} (Section~\ref{sec:pattern_detection}). 

\item These traffic patterns are grouped together in {\bf Anonymity Set Generation} (Section~\ref{sec:AS_gen}), wherein local parameters are identified for each anonymity set. 

\item {\bf Early Anonymity Set Detection} (Section~\ref{sec:early_pred}) tells the defense when to switch from global to local parameters based on packets of the observed live trace. 
\end{enumerate}

\subsection{Intra-Webpage Pattern Detection}
\label{sec:pattern_detection}

Webpages do not generate a single, uniform traffic pattern; rather, they produce a number of patterns. This can be due to: dynamically generated content such as personalized recommendations~\cite{ravi2009survey}, advertisements that can differ significantly in size and frequency~\cite{veverka2024impact}, and variations in localization tailored for different countries~\cite{okonkwo2023localization}. As an illustrative example, Figure~\ref{fig:tam_variation} shows four traces from the same webpage in the Sirinam et al.\ dataset \cite{sirinam2018deep}, visualized using the Traffic Aggregation Matrix (TAM)~\cite{shen2023subverting}.  The two rows highlight two distinct recurring structures, highlighting that even a single page can yield
multiple characteristic patterns.

\begin{figure}[ht]
    \centering
    \subfloat[Pattern 1]{%
        \includegraphics[width=0.50\textwidth]{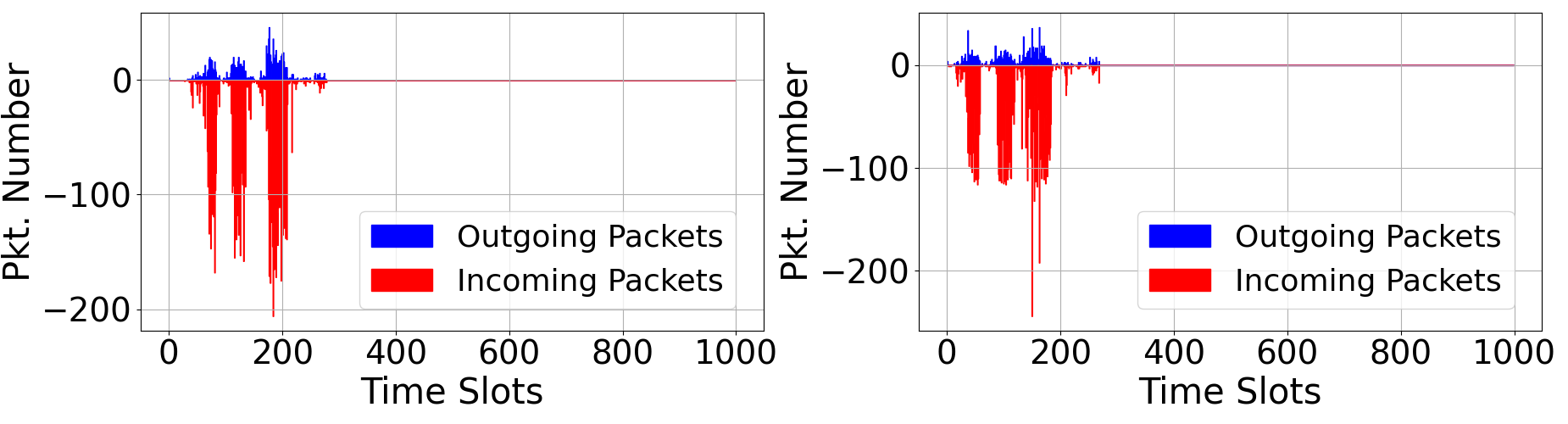}%
        \label{fig:website70}%
    }
    \hfill
    \subfloat[Pattern 2]{%
        \includegraphics[width=0.50\textwidth]{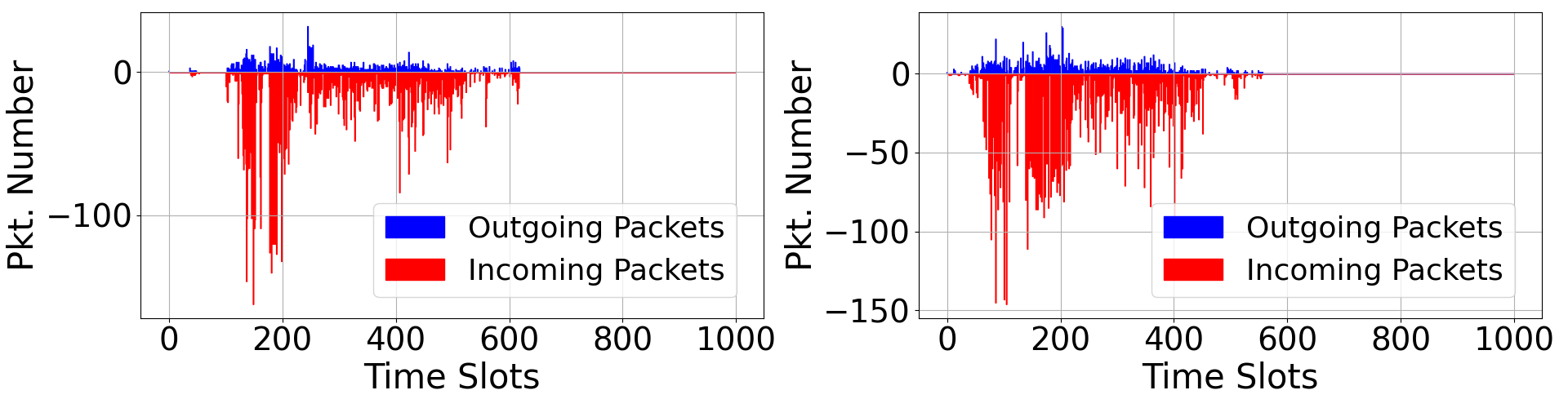}%
        \label{fig:website15}%
    }
    \caption{Visualization of the TAMs of four traces from the \emph{same} page, in the dataset obtained by \cite{sirinam2018deep}. Each TAM divides the first 1000 time-slots (80 ms per slot) into rows for outgoing packets (blue, positive values) and incoming packets (red, negative values); the height of each bar records the packet count in that slot. The first two traces (top row) exhibit a very similar pattern, while the next two traces (bottom row) share a distinct structure. This highlights the possibility of multiple recurring traffic patterns within a single page.}
    \label{fig:tam_variation}
\end{figure}

As a result of these observations, aggregating all traces from a single webpage into one homogeneous profile is inefficient and often leads to over-padding and degraded performance. However, many existing defenses construct anonymity sets at the webpage level, assuming that all traces from the same page behave similarly~\cite{wang2014effective,nithyanand2014glove,shen2024real}.
Instead, our approach aggregates traffic at the \textit{pattern level} to form smaller, more homogeneous groups. 
% In practice, however, as we mentioned, a single page can generate multiple distinct traffic profiles.

To obtain distinct traffic patterns within each webpage, we first represent each network trace as a time series using the TAM representation introduced in ~\cite{shen2023subverting}. In our TAM representation, we divide the total page load time into a fixed number of time slots and record the number of incoming and outgoing packets in each slot. This results in a two-dimensional time series, where each trace can be viewed as a bivariate sequence capturing the dynamics of incoming and outgoing traffic over time.

With traces represented as TAM time series, we cluster the traces of each webpage individually to isolate recurring traffic patterns. We refer to the resulting clusters as \emph{intra-webpage clusters} (i.e., for the remainder of this section, when we refer to intra-webpage clusters, we mean clusters obtained by grouping traces originating from a single webpage).

%  To perform this task, we adopt the \emph{Cluster Affinity Search Technique} (CAST)~\cite{aghabozorgi2014hybrid}. CAST does not require us to specify the number of clusters in advance; instead, it uses an \emph{affinity threshold} $T$ (which is a hyperparameter) to determine membership. The affinity of a time series \(F_x\) with respect to a particular intra-webpage cluster \(\mathcal{S}\) is defined as follows:

% \[
% a_{\mathcal{S}}(F_x) \;=\; \frac{1}{\lvert \mathcal{S}\rvert} \sum_{F_y \,\in\, \mathcal{S}} A(F_x, F_y).
% \]
% \noindent
% where \(A(F_x, F_y)\) is the similarity between time series \(F_x\) and \(F_y\), and \(\lvert \mathcal{S}\rvert\) is the size of the intra-webpage cluster. CAST evolves a cluster by repeatedly selecting the unassigned time series with the highest affinity to the cluster:

% \begin{itemize}
% \item If its affinity exceeds the fixed threshold $T$, the series is added to the intra-webpage cluster.
% \item Otherwise, the algorithm checks if any existing member in the intra-webpage cluster has an affinity falling below the threshold; if so, that low-affinity member is removed.
% \item If no member meets the removal criterion, the intra-webpage cluster is deemed complete (closed), and CAST proceeds to form a new intra-webpage cluster with the remaining data.
% \end{itemize}

To perform this clustering, we adopt the \textit{Cluster Affinity Search Technique} (CAST) \cite{aghabozorgi2014hybrid}. CAST is a similarity-based clustering algorithm that avoids pre-specifying the number of clusters. Instead, it incrementally builds clusters by comparing each trace's \textbf{affinity} (its average similarity) to the current cluster. We define the affinity of a time series $F_x$ to a cluster $S$ as:
\begin{equation}
a_S(F_x) = \frac{1}{|S|} \sum_{F_y \in S} A(F_x, F_y),
\end{equation}

\noindent
where $A(F_x, F_y)$ is the similarity measure between two TAM series. Intuitively, $a_S(F_x)$ captures \textit{how well $F_x$ fits with the existing members of $S$}: high affinity indicates that $F_x$ behaves similarly to most traces in $S$, while low affinity indicates that it is not consistent with the dominant traffic pattern in that cluster. CAST maintains a currently ``open'' cluster and evolves it using a fixed affinity threshold $T$. At each step, the unassigned trace with the highest affinity to the open cluster is examined:

\begin{itemize}
    \item \textbf{Addition:} If its affinity is at least $T$, the trace is added to the cluster.
    \item \textbf{Removal:} Otherwise, CAST checks if any current member of the cluster has affinity below $T$ with the current cluster. Such traces are considered weakly related to the cluster and are removed.
    \item \textbf{Closure:} If no addition or removal is possible, the cluster is considered \textit{stable} and is closed; CAST then begins forming the next cluster from the remaining traces.
\end{itemize}

%While CAST’s flexibility in not requiring an \emph{a priori} number of clusters is advantageous, 
 Our preliminary evaluations (discussed in Appendix \ref{appendix:cast}) revealed that for traces in a single webpage, CAST generated several large intra-webpage clusters, along with a long tail of smaller ones. Small intra-webpage clusters are not suitable, as our goal is to regularize all traces within an intra-webpage cluster to appear identical; if an intra-webpage cluster contains too few traces, the defense incurs the overhead of padding and shaping with minimal privacy benefit. To be more precise, the attacker’s uncertainty remains low despite the added cost, resulting in an inefficient overhead-to-privacy tradeoff.
 % lack the structural stability needed for effective anonymity set construction
%, while overly large clusters dilute intra-set similarity. 
%Both effects undermine subsequent regularization steps by inflating overhead and weakening pattern alignment.
 To produce better clusters, we introduced four modifications to the original CAST algorithm proposed in \cite{aghabozorgi2014hybrid}, which we will describe below.
 %to produce more balanced and cohesive clusters that align better with the diverse traffic modes observed within a single webpage.

\textbf{Affinity Computation.}  
As we mentioned, the affinity of a time series \(F_x\) with respect to a candidate cluster \(\mathcal{S}\) is computed as the average of the pairwise similarity scores between \(F_x\) and each member \(F_y \in \mathcal{S}\).
Following prior work~\cite{wu2018random,badiane2022empirical}, the similarity score \(A(F_x, F_y)\) is computed using an exponential function of the squared Euclidean distance:
\[
A(F_x, F_y) \;=\; \exp\!\Bigl(-\frac{d^2(F_x, F_y)}{\sigma^2}\Bigr),
\]
where \(\sigma\) is a bandwidth hyperparameter that governs the sensitivity of the similarity function. Previous studies~\cite{correa2012locally} have shown that a global \(\sigma\) can fail to capture variations in data density. %Some regions may be densely packed, while others are sparse, requiring different scales for accurate comparison.

To address this, we adopt a local scaling approach inspired by self tuning spectral clustering~\cite{ng2001spectral,zelnik2004self}, in which each time series \(F_x\) is assigned its own scale parameter \(\sigma_x\), typically set to the distance between \(F_x\) and its \(K\)-th nearest neighbor. 
% (we found \(K = 7\) to perform well). 
The similarity is then computed as:
\[
A(F_x, F_y) \;=\; \exp\!\Bigl(-\frac{d^2(F_x,F_y)}{\sigma_x \, \sigma_y}\Bigr).
\]

\textbf{Cleaning Step.}
After forming the initial clusters, we perform a final pass. %to refine cluster memberships. 
For each point \(c\) in each cluster \(\mathcal{C}_i\), we compute its affinity to every other cluster \(\mathcal{C}_j\). If we find a cluster \(\mathcal{C}_j\) in which the point \(c\) achieves a strictly higher affinity than in its current cluster \(\mathcal{C}_i\), we remove \(c\) from \(\mathcal{C}_i\) and reassign it to \(\mathcal{C}_j\). This process is repeated until no point changes its cluster membership. %or until a maximum number of iterations (20 in our case) is reached. The purpose of this “cleaning” procedure is to correct any misassignments that might have occurred during the main clustering phase, ensuring that each point ends up in the cluster where it has the highest average similarity to the other members. This step thus sharpens cluster boundaries and increases overall cluster homogeneity.

\textbf{Post-Processing Step.}
To limit the number of clusters produced by our modified CAST algorithm, we adopt a post-processing procedure inspired by \cite{tyuryukanov2017post}.  
Each cluster consists of a set of TAM-based trace representations. We measure its \emph{cut} and \emph{volume} using the similarity score \(A\) introduced earlier:
\[
\mathrm{cut}(\mathcal{C}) \;=\; \sum_{f \in \mathcal{C}} \sum_{f' \notin \mathcal{C}} A(f,f'),
\quad
\mathrm{vol}(\mathcal{C}) \;=\; \sum_{f\, ,\,f' \in \mathcal{C}} A(f,f').
\]
The \emph{expansion ratio}
$\phi(\mathcal{C}) = \mathrm{cut}(\mathcal{C})/\mathrm{vol}(\mathcal{C})$
is low when traces inside \(\mathcal{C}\) are highly similar while remaining well separated from traces outside.%, which reduces overhead in the regularization phase.

We then apply an iterative post-processing routine: as long as the number of clusters exceeds a predefined threshold, we sort the clusters by size, select the smallest cluster, and merge it with a neighboring cluster such that the resulting partition minimizes the largest expansion ratio among all clusters~\cite{sanchez2014hierarchical}. 
%This strategy gradually reduces the number of clusters while improving their cohesion and balance, avoiding both excessively small and overly heterogeneous groupings.

\textbf{Dynamic Affinity Threshold.}
Finally, rather than using a fixed affinity threshold for all traces, we compute it in a data-driven manner. We first calculate the global mean similarity across all pairs of traces:
\[
T \;=\; \frac{1}{n^2}\sum_{x < y \le n} A(F_x, F_y),
\]
where \(A(F_x, F_y)\) is the pairwise similarity between traces \(x\) and \(y\), and $n$ is the total number of traces. This global average \(T\) serves as a baseline for typical similarity values within the dataset. We then set the CAST affinity threshold to \(T\) (i.e., the minimum similarity required to add a new trace to an intra-webpage cluster).
This approach ensures that the threshold automatically adapts to the overall range of similarities in the dataset, rather than being a manually tuned constant hyperparameter, thereby improving the robustness of the resulting clusters.

%In summary, the modified CAST procedure, with its adaptive local scaling, iterative cleaning, and cluster merging post processing, produces balanced, cohesive groups of TAM traces. 
We apply the modified CAST algorithm separately to each webpage in the dataset to extract a set of characteristic traffic patterns. These extracted patterns form the foundation for all subsequent phases. 

% The complete pseudocode for this phase is given in Appendix~\ref{appendix:cast}.

\subsection{Anonymity Set Generation}
\label{sec:AS_gen}
Once we have extracted distinct traffic patterns from each webpage, as described in Section~\ref{sec:pattern_detection}, we proceed to cluster these patterns into anonymity sets. This phase organizes structurally similar patterns, regardless of their originating web pages, into groups that will later be reshaped using a shared regularization policy. In contrast to prior work that clusters at the webpage level~\cite{wang2014effective,nithyanand2014glove,shen2024real}, this fine-grained grouping reduces overhead and improves generalization.

Our anonymity sets are constructed to satisfy two key properties that underlie the provable security guarantee offered by our defense:

\begin{itemize}
    \item \textbf{k-anonymity:} Each anonymity set must contain at least $k$ distinct traffic patterns~\cite{byun2007efficient}. In this way, even if an adversary correctly identifies the anonymity set to which a trace belongs, they are left with at least $k$ possibilities, bounding the probability of correctly inferring the true pattern to at most $1/k$.
    \item \textbf{l-diversity:} The patterns within each anonymity set must originate from at least $l$ different webpages~\cite{machanavajjhala2007diversity}. This criterion ensures that even if the anonymity set is exposed, the attacker cannot infer the correct webpage simply because all patterns come from a single source.
\end{itemize}

For this phase, we adopt the $k$-anonymity-based clustering algorithm introduced in \textit{Palette}~\cite{shen2024real} as our base algorithm. In \textit{Palette}, anonymity sets are generated at the webpage level. The algorithm first computes the Traffic Aggregation Matrix for each trace of a webpage and then aggregates these into a \emph{super-matrix} by an element-wise maximum. These super-matrices are then used as the elements (i.e., the representatives of each webpage) for a $k$ anonymity-based clustering algorithm~\cite{byun2007efficient}, with the similarity between two webpages measured by computing the Euclidean distance between their corresponding super-matrices. The clustering itself follows a greedy strategy: it constructs one anonymity set at a time by first selecting an initial webpage and then repeatedly adding the webpage whose super-matrix is closest, based on Euclidean distance, to the current super-matrix of the partially constructed set. Once $k$ webpages have been assigned, the process repeats to form the next set. The output is a list of anonymity sets, each containing at least $k$ webpages. We adopt the same $k$ anonymity-based clustering algorithm as used in \textit{Palette}, but with two key modifications:

\textbf{Pattern-level clustering.} Our method operates at the pattern level rather than the webpage level. As explained in Section \ref{sec:pattern_detection}, a single webpage can produce multiple distinct traffic patterns due to various reasons; clustering at this finer granularity yields more homogeneous groups.  As an illustrative example, Figure~\ref{fig:pattern_vs_website} compares the impact of clustering at the webpage level versus the pattern level. For each fixed value of $k$, we conduct two separate experiments. In the first, we compute one supermatrix per webpage by taking the element-wise maximum over all TAMs from that webpage, and then apply the clustering algorithm from \textit{Palette} to these webpage-level supermatrices. In the second, we repeat the process at the pattern level by computing one supermatrix per extracted traffic pattern and clustering those instead, again using \textit{Palette}. In both cases, after clustering, we treat each cluster’s supermatrix as the defended version of all traces in the cluster and compute the bandwidth overhead across all the traces. This provides a proxy for how well the clustering captures homogeneity. As shown in the figure, clustering at the pattern level consistently yields lower overhead.

% \begin{figure}[ht]
%     \centering
%     \subfloat[Website-level Clustering Overhead]{%
%         \includegraphics[width=0.50\textwidth]{Figures/palette/pallete_with_website.png}%
%         \label{fig:website_overhead}%
%     }
%     \hfill
%     \subfloat[Pattern-level Clustering Overhead]{%
%         \includegraphics[width=0.50\textwidth]{Figures/palette/pallete_with_pattern.png}%
%         \label{fig:pattern_overhead}%
%     }
%     \caption{Comparison of total data overhead when using \textit{Palette} on the website level versus the pattern level. The results show that clustering at the pattern level leads to a significant reduction in overhead, indicating more effective regularization due to the increased homogeneity of patterns.}
%     \label{fig:pattern_vs_website}
% \end{figure}

\begin{figure}[ht]
    \centering
    \includegraphics[width=0.48\textwidth]{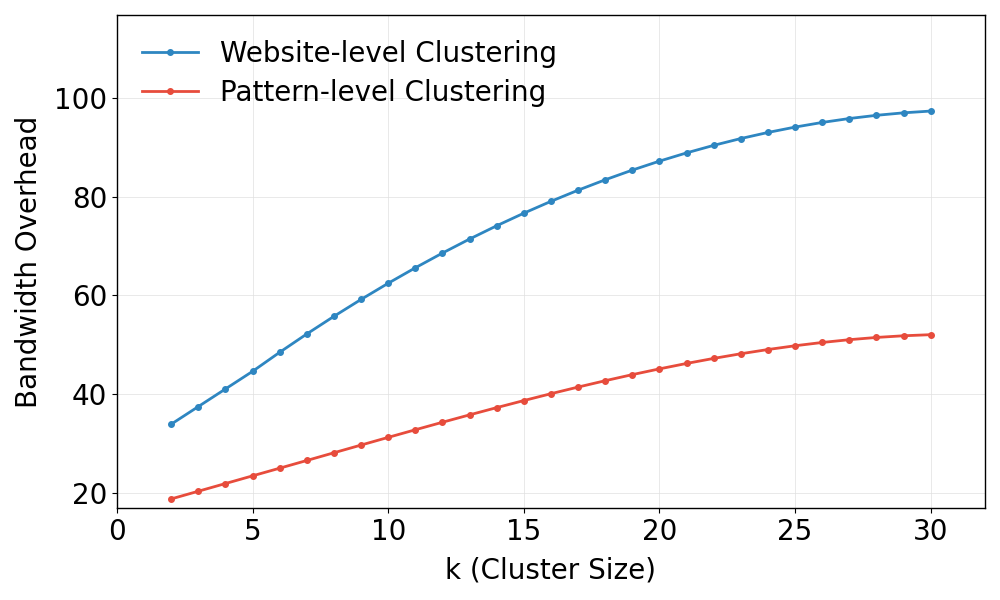}
    \caption{
    Comparison of website-level vs. pattern-level clustering. For each value of $k$, clustering is applied to supermatrices constructed at the website or pattern level, and the resulting average bandwidth overhead is measured. 
    Pattern-level clustering consistently results in lower overhead, especially as $k$ increases, indicating that clustering finer-grained traffic patterns captures homogeneity more effectively than aggregating at the website level.
    }

    \label{fig:pattern_vs_website}
\end{figure}

\textbf{Diversity-aware distance metric.} Instead of computing the Euclidean distance between
super-matrices, we define a distance function $d$ that directly captures attacker success after regularization. Given an anonymity set $C$ currently under construction and a candidate pattern $p$ (i.e., an intra-webpage cluster), we evaluate the impact of merging them by computing the average-case attacker success rate over the combined set $C' = C \cup p$.

For each pair of regularization parameters $(p_{\text{in}}, p_{\text{out}})$ in a predefined grid $\mathcal{P}$, we apply Tamaraw to all traces in $C'$ and compute the non-uniform attacker accuracy:

\[
  \bar{A}\bigl(C';p_{\text{in}},p_{\text{out}}\bigr)
  \;=\;
  \sum_{\ell}
    \frac{|C'_{\ell}|}{|C'|}
    \cdot
    \frac{
      \displaystyle
      \max_{w}\bigl|\{\,t\in C'_{\ell}:\mathrm{site}(t)=w\}\bigr|
    }{
      |C'_{\ell}|
    },
\]
\noindent
where $C'_{\ell}$ denotes the subset of traces in $C'$ whose regularized lengths equal $\ell$. This measures the expected success rate of an optimal attacker who predicts the most frequent webpage label in each length bucket.

The distance $d(C, p)$ is then defined as the average attacker accuracy over the entire grid:

\[
  d(C,p)\;=\;
  \frac{1}{|\mathcal{P}|}
  \sum_{(p_{\text{in}},p_{\text{out}})\in\mathcal{P}}
    \bar{A}\bigl(C';p_{\text{in}},p_{\text{out}}\bigr).
\]

This formulation implicitly enhances $l$-diversity, since adding a pattern from a new webpage typically reduces the attacker’s success rate based on the webpage in the majority, thus lowering $d(C,p)$. We provide empirical validations for this attribute in Appendix~\ref{appendix:cluster_diversity}. During clustering, we therefore always merge the candidate $p$ that minimizes $d(C,p)$, encouraging both homogeneity in regularized shape and diversity in origin.  Algorithm~\ref{alg:pattern_clustering} provides detailed pseudocode of our proposed anonymity set generation algorithm.  %our method ensures that each anonymity set contains at least \(k\) patterns drawn from at least \(l\) distinct webpages. 
%These anonymity sets will serve as the basis for our global-to-local regularization defense strategy.

\begin{algorithm}[htb]
\caption{Pattern-Level Anonymity Set Generation}
\label{alg:pattern_clustering}
\textbf{Input:}\\
\quad $P = \{p_1,\dots,p_N\}$: all extracted patterns\\
\quad $k$: anonymity parameter\\
\quad $d(S,p)$: distance function between anonymity set $S$ and pattern $p$ \\
\textbf{Output:}\\
\quad $\mathcal{S}=\{S_1,\dots,S_m\}$: final anonymity sets
\begin{algorithmic}[1]
\State $\mathcal{S} \gets \varnothing$ \Comment{Anonymity sets}
\State $U \gets P$ \Comment{the set of unassigned patterns}
\Statex \textbf{Step 1: Seed the first set}
\State Choose any $p_{\text{seed}} \in U$
\State $S_1 \gets \{p_{\text{seed}}\}$; $\mathcal{S} \gets \{S_1\}$; 
       $U \gets U \setminus \{p_{\text{seed}}\}$
\Statex \textbf{Step 2: Build additional sets}
\For{$i = 1$ to $\left\lfloor \frac{|P|}{k} \right\rfloor$}
    \If{$|S_i| < k$}
        \While{$|S_i| < k$ \textbf{and} $U \neq \varnothing$}
            \State $p^\star \gets \arg\min_{p \in U} d(S_i, p)$
            \State $S_i \gets S_i \cup \{p^\star\}$;
            $U \gets U \setminus \{p^\star\}$
        \EndWhile
    \EndIf
    \If{$U \neq \varnothing$}
        \State $p_{\text{new}} \gets \arg\max_{p \in U} \sum_{S \in \mathcal{S}} d(S, p)$
        \State Create $S_{i+1} \gets \{p_{\text{new}}\}$;
        $\mathcal{S} \gets \mathcal{S} \cup \{S_{i+1}\}$;
        $U \gets U \setminus \{p_{\text{new}}\}$
    \EndIf
\EndFor
\Statex \textbf{Step 3: Assign remaining patterns}
\ForAll{$p \in U$}
    \State $j \gets \arg\min_{S \in \mathcal{S}} d(S, p)$
    \State $S_j \gets S_j \cup \{p\}$
\EndFor
\State \Return $\mathcal{S}$
\end{algorithmic}
\end{algorithm}

The two previous phases (pattern extraction and anonymity set generation) comprise the offline training stage of our framework. A high-level illustration of this process is shown in Figure~\ref{fig:overview}, which depicts how traffic traces from different webpages are transformed into anonymity sets based on structured patterns.

\begin{figure}[ht]
    \centering
    \includegraphics[width=0.48\textwidth]{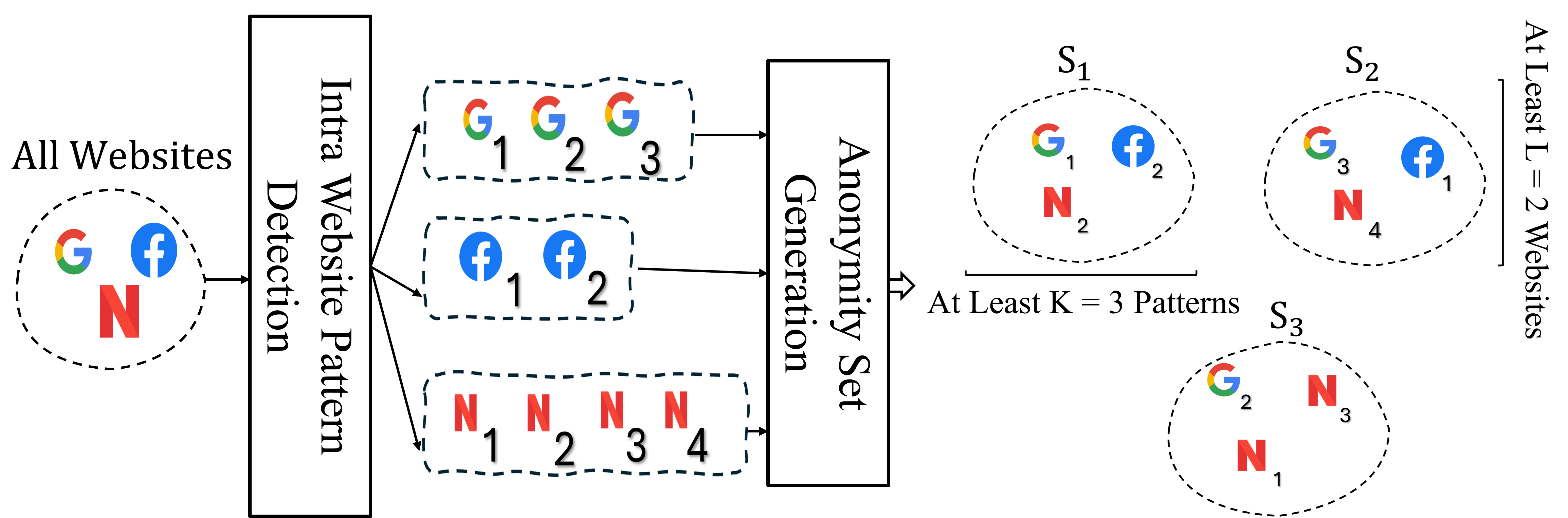}
    \caption{High–level workflow of the first two phases of our defense.  
  \textbf{1.~Pattern extraction.}  
  For each webpage in the training dataset, we group its traces into a small number of stable, recurring traffic patterns (dashed boxes), reflecting variability due to CDNs, localization, and user behavior.  
\textbf{2.~Anonymity set construction.}  
The extracted patterns are then clustered across different webpages to form anonymity sets. In this example, each set contains at least $k = 3$ distinct patterns originating from at least $l = 2$ different webpages, thereby satisfying $k$-anonymity and $l$-diversity. A lightweight, cluster-specific regularization schedule is precomputed for each set.}
    \label{fig:overview}
\end{figure}

% By extending the base methodology of \textit{Palette} to operate at the pattern level and by employing a customizable distance function \( d(C,p) \) augmented with diversity considerations, our anonymity set generation process provides a flexible and robust framework. This approach reduces regularization overhead and increases the applicability of the defense in open-world scenarios, where users often visit webpages not included in a predetermined dataset, while still offering provable security guarantees.

\subsection{Early Anonymity Set Detection}
\label{sec:early_pred}

% Our defense begins every connection under the global regularization
% scheme (Sec.~\ref{sec:framework}).  To cut overhead as soon as it is
% safe, we must decide—\emph{during} the page load—which anonymity set
% (AS) the live trace belongs to and when to switch to that set’s lighter
% parameters.  We cast this as an \textbf{early time–series
% classification} problem, adapting the Reliable Early Classification of
% Time Series (ECDIRE) framework \cite{mori2017reliable}. ECDIRE learns, for each class, the earliest \emph{safe timestamp} at
% which the class can be predicted with a user‑chosen fraction
% \(\alpha\) of its full–length accuracy.  At run time the classifier
% is only queried at those safe timestamps, eliminating needless tests
% at earlier instants and guaranteeing that a decision is made as soon
% as the desired confidence level is achieved.

As we mentioned, our defense follows a \emph{global-to-local} regularization strategy.  
Every new page load is protected by singular global regularization parameters until enough packets have arrived to identify the anonymity set that best matches the live prefix; once the set is known, we switch to lighter, local regularization parameters.

% We tune global regularization by searching its parameter space for the best global overhead-privacy trade-off, an approach commonly adopted in prior defenses such as \cite{dyer2012peek,cai2014systematic,holland2020regulator}.  
% This tuning process uses all traces in the training set.
%Using all traces in the training set, we first derive a set of \emph{global} parameters, giving a one-size-fits-all set of parameters that aim to reshape the incoming traffic.
% We then repeat the same search \emph{inside each anonymity set}, producing lighter local parameters that are tailored to the traffic patterns of that set (a concrete example appears in Section~\ref{sec:exp_setup}). 
%In our setting, we realistically assume that the user cannot know in advance which traffic pattern will characterize the current connection. To address this uncertainty, our defense monitors the evolving time series features of the trace during the initialinitial phase of the connection and tries to assign them to their corresponding anonymity set. As soon as the classifier can predict the corresponding anonymity set with high confidence, the system switches from the global parameters to those optimized for the predicted set. 

For the global-to-local switch, deciding \emph{which} set to choose and \emph{when} it is safe to switch is an instance of \textbf{early time-series classification}, a task where the goal is to make accurate predictions based on incomplete sequences, as early as possible.  
Early classification has been studied extensively in latency-critical domains such as medical diagnosis~\cite{he2014early,dachraoui2015early} and industrial process monitoring~\cite{lv2019effective,ahn2020development}.
%, where an incomplete time series must be labelled as soon as possible while still meeting a prescribed accuracy target.  
In this work, we adopt the ECDIRE framework~\cite{mori2017reliable}, which has been designed for reliable early classification of time series. 
ECDIRE learns, for every class (in our case anonymity set), the earliest \emph{safe timestamp} at which the class can be predicted with at least a user-selected fraction~$\alpha$ of the accuracy that would be achieved on complete traces.  
During deployment the classifier is queried only at those safe timestamps, which removes unnecessary tests on very short prefixes and guarantees that a decision is issued as soon as the chosen confidence threshold is met.

For every anonymity set \(\mathcal{S}\) we compute a safe timestamp
\(\tau_{\mathcal{S}}\) as the earliest prefix length (i.e., time) \(t\) at which a
validator (described below) can distinguish \(\mathcal{S}\) from all
other sets with accuracy \(\ge\alpha\,A_{\mathcal{S}}^{\text{full}}\),
where \(A_{\mathcal{S}}^{\text{full}}\) is the accuracy of the validator on full traces of \(\mathcal{S}\).  We never test
\(\mathcal{S}\) before \(\tau_{\mathcal{S}}\); conversely, if the trace
is \emph{not} accepted at \(\tau_{\mathcal{S}}\), 
\(\mathcal{S}\) is never considered again for that trace.
This “single‑shot’’ rule forces every trace that joins
\(\mathcal{S}\) to switch defenses at the same fixed time, closing a
potential timing side channel.

In the original ECDIRE algorithm~\cite{mori2017reliable}, the training process initially involves training a separate classifier at each timestep. Once training is complete, a post-processing step identifies the safe prediction timestamps, the points in time where early classification is both accurate and confident. The final model then retains only the classifiers corresponding to these selected timestamps, discarding the rest.
In our case, WF classifiers are deep networks, so to train and retain a full model for each timestamp would be prohibitively expensive in time and memory.
Instead we factor the prediction task into a single deep backbone plus a collection
of lightweight per-site models as follows.

\begin{enumerate}[leftmargin=13pt,itemsep=2pt]
\item \textbf{Stage A – webpage predictor (Holmes).}  
      We train one \emph{Holmes} network~\cite{deng2024robust}, a spatial–temporal
CNN encoder learned with supervised contrastive loss so that partial
traces embed close to their full counterparts.
At inference time Holmes outputs the most likely webpage \(w\) given an incomplete input trace.

\item \textbf{Stage B – Pattern predictor (per‑site k‑fingerprinting).}  
      For each webpage \(w\) and each time stamp \(t\) of interest, we train a small k fingerprinting (kFP) random forest classifier~\cite{hayes2016k}, denoted \(\text{kFP}_{w}^{(t)}\), using the traces of the site truncated at time \(t\) and the labels of each trace being the number of the cluster intra-site (i.e. pattern) of that trace. kFP requires only a modest number of examples, making it well-suited for scenarios where per-site data is limited.

\end{enumerate}

Given an incoming trace prefix of length \(t\), the prediction process proceeds as follows:
\begin{enumerate}[leftmargin=13pt,itemsep=2pt]
  \item Holmes predicts the most likely webpage \(w\).
  \item  The corresponding kFP model \(\text{kFP}_{w}^{(t)}\) identifies the pattern \(p\).
  \item By construction, the pair \((w, p)\) uniquely determines an anonymity set \(\mathcal{S}_{w,p}\). If \(t = \tau_{\mathcal{S}_{w,p}}\) (i.e., the current timestamp is the \emph{safe prediction timestamp} for that anonymity set), we immediately switch to that set’s lighter parameters; otherwise, we continue applying the global regularization schedule.
\end{enumerate}

After training the single Holmes model and all \(\text{kFP}_{w}^{(t)}\) classifiers across timestamps, we determine a safe timestamp \(\tau_{\mathcal{S}}\) for each anonymity set \(\mathcal{S}\), indicating the earliest point at which a confident classification can be made. For each webpage \(w\), we keep only the \(\text{kFP}_{w}^{(t)}\) models corresponding to its safe timestamps. The Holmes model is site-agnostic and trained once, then reused across all traces and timesteps. This design preserves ECDIRE’s early-decision capability while reducing computational cost by avoiding deep model training at every time step. Figure~\ref{fig:early_pred} illustrates this switching process and its alignment with safe timestamps.

%Holmes provides a fast, robust webpage prediction, while the lightweight kFP forests perform fine-grained pattern classification with minimal overhead. Together, they ensure that only one deep model and a few small classifiers are evaluated per trace, keeping online cost low. 

% Given an incoming trace, let \(w\)  be the predicted webpage of this trace (obtained by the webpage predictor) and \(p\) be the pattern returned by the kFP corresponding to webpage \(w\) at the current safe timestamp (pattern predictor).  By construction, the pair
% \((w,p)\) is unique in our training table and is placed in a single
% anonymity set; we denote that set by
% \(\mathcal{S}_{w,p}\).  Formally,

% \[
% \text{trace prefix}\;t
% \;\;\xrightarrow{\;\;\text{Holmes}\;\;}\;w
% \;\;\xrightarrow{\;\;\text{kFP}_{w}^{(t)}\;\;}\;p
% \;\;\longmapsto\;\mathcal{S}_{w,p}.
% \]

% When the trace reaches
% \(t=\tau_{\mathcal{S}_{w,p}}\) and the cascade outputs \((w,p)\), we
% switch from the global regularization to the lighter parameters
% associated with \(\mathcal{S}_{w,p}\).  
% Holmes supplies an early, defense‑robust webpage cue, while the compact
% kFP forests separate patterns with negligible computation; together
% they allow us to evaluate only one deep model per trace and a handful
% of tiny forests, keeping on‑line cost low.

% ------------------------------------------------------------------
% Wide figure that spans both columns in a two-column layout
% (e.g., IEEEtran conference, ACM sigconf, etc.)
% ------------------------------------------------------------------
\begin{figure}[t]
    \centering
    % Replace the filename with the actual path of the exported diagram.
    \includegraphics[width=\linewidth]{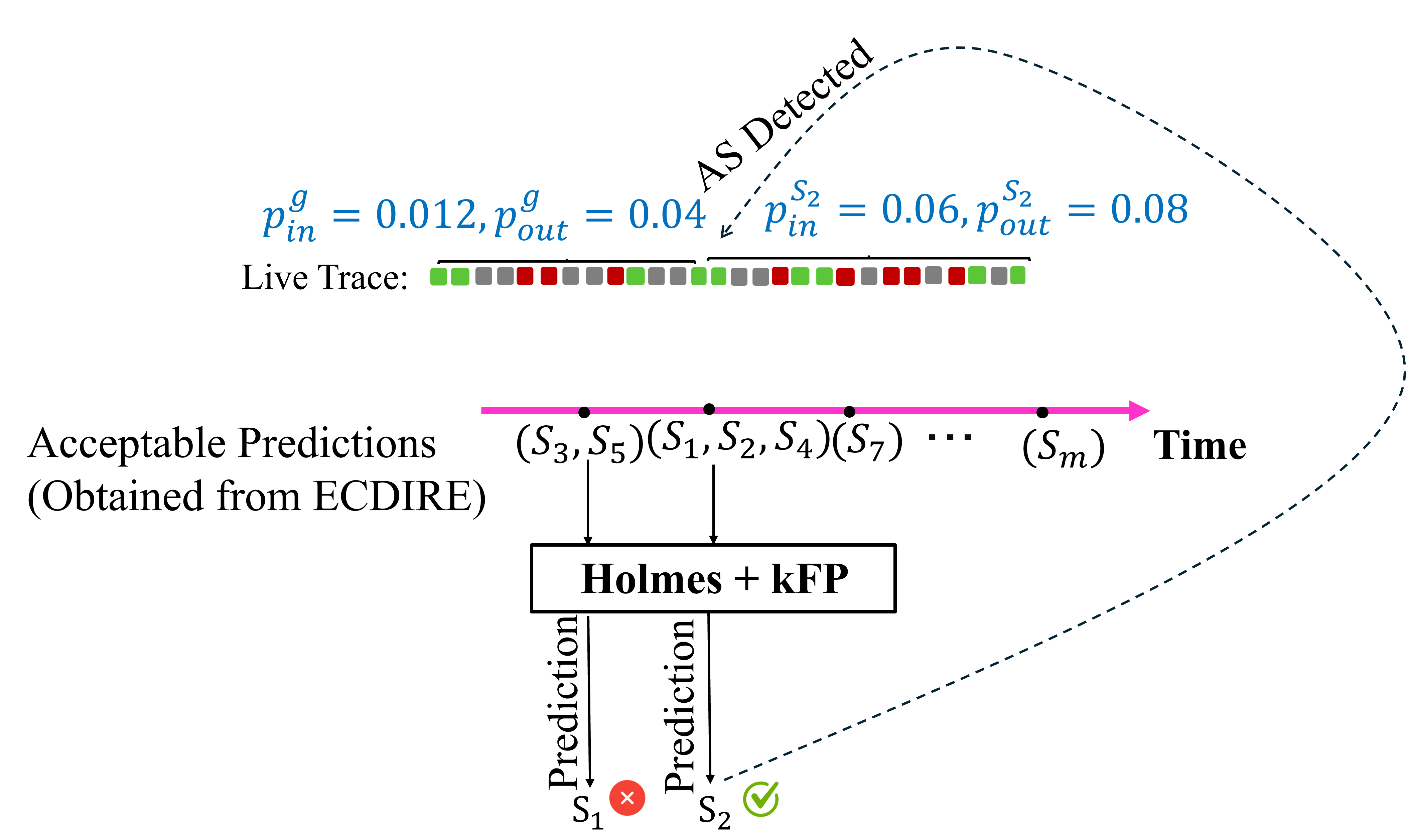}
   \caption{
Illustration of early anonymity set prediction and parameter switching.
An incoming trace is initially regularized using the global parameters \((p_{\text{in}}^g = 0.012, p_{\text{out}}^g = 0.04)\) in this example.
At each predefined safe timestamp, the classifier attempts to assign the trace to one of the candidate anonymity sets. 
In the first attempt, the classifier predicts \(S_1\), but no switch occurs because \(S_1\) is not valid at that timestamp. 
At the second safe timestamp, the trace is matched with \(S_2\), which is an acceptable set at that point. 
This triggers a transition to \(S_2\)'s per-set lighter parameters \((p_{\text{in}}^{S_2} = 0.06, p_{\text{out}}^{S_2} = 0.08)\), which are applied for the rest of the connection. 
Each anonymity set is tied to a unique safe timestamp, so switching can occur at most once.
}
\label{fig:early_pred}

\end{figure}

% To evaluate the general framework introduced in this paper, we embed it in a concrete, well-studied defence: Tamaraw \cite{cai2014systematic}.  Tamaraw already offers an information-theoretic guarantee, but it enforces a single constant send rate for all traffic, which inflates bandwidth and latency for many pages. 
% Classic Tamaraw transmits fixed-size cells of 750 bytes at two constant rates:  
% $\rho_{\text{out}}$ seconds between successive upstream cells and $\rho_{\text{in}}$ seconds between successive downstream cells.  
% Padding continues in each direction until its cell count is the next multiple of a quantum $L$.  
% Because a single triple $(\rho_{\text{out}},\rho_{\text{in}},L)$ must suit every site, many pages carry excessive dummy traffic.  
% Adaptive Tamaraw removes this one-size-fits-all constraint by dividing each connection into two phases.  

%-----------------------------------------------------------
% 6.X  Security Bound for Adaptive Tamaraw
%---------------------------------------------------------
\subsection{Security Bound for \textit{Adaptive Tamaraw}}
\label{sec:security-bound}

Despite transitioning from a fixed global defense to lighter, AS-specific configurations after a safe prediction point, Adaptive Tamaraw retains a provable security guarantee. Specifically, we show that the defense remains \emph{non-uniformly weighted} \(\delta\)-non-injective, ensuring that the attacker’s average success probability is formally bounded. This bound reflects the attacker's expected success probability across the entire distribution of defended traffic, aligned with standard metrics from prior work, including Tamaraw ~\cite{cai2014systematic}, which included non-uniform (i.e., average) security.

\begin{theorem}[Global Non-Uniformly Weighted \(\delta\)–Non-Injectivity]
\label{thm:delta-global-maintext}
Let \(\mathcal{S}\) be the set of anonymity sets constructed in Section~\ref{sec:AS_gen}. For a fixed global regularization parameter pair \((p_{\text{in}}, p_{\text{out}})\), let \(\bar{A}(S_i; p_{\text{in}}, p_{\text{out}})\) denote the expected attacker success rate over anonymity set \(S_i \in \mathcal{S}\), as defined in Section~\ref{sec:max_accuracy} (Eq.~5). Define the global injectivity parameter \(\delta\) as:
\[
\frac{1}{\delta} = \mathbb{E}_{S_i \sim \mathcal{S}} \left[\bar{A}(S_i; p_{\text{in}}, p_{\text{out}})\right].
\]
Then Adaptive Tamaraw is non-uniformly weighted \(\delta\)-non-injective, and the attacker’s average success probability is bounded by:
\[
\Pr[\text{success}] \;\le\; \frac{1}{\delta}.
\]
\end{theorem}

The formal proof of this result is presented in Appendix~\ref{appendix:proof}.

This result is intuitive: the attacker's overall bounded success rate is the average across the bound for each combination of anonymity set and trace length. 
Any two traces from the same set and of the same length are made indistinguishable by the defense, a principle inherited from Tamaraw. Consequently, the attacker’s optimal strategy in an anonymity set reduces to guessing the most common webpage in each such length-matched group, which limits their average success rate captured by \(\bar{A}(S_i)\). Our novelty lies in deriving the overall security bound by averaging these individual success probabilities across all possible anonymity sets.

%% file: Sections_revision/7_experiments.tex
\section{Experimental Evaluation}
\label{sec:experiments}

This section empirically evaluates the effectiveness, efficiency, and generalizability of Adaptive Tamaraw. The goal is to show that the defense lowers bandwidth and latency overhead while maintaining provable bounds on attacker success, even under strong adversaries. Section~\ref{sec:exp_setup} outlines the experimental setup; Section~\ref{exp:closed_world} tests performance on in-training webpages (defended sites seen during training); Section~\ref{exp:open_world} examines robustness on out-of-training webpages with unseen traffic; and Section~\ref{sec:max_accuracy} checks the tightness of theoretical bounds from Section~\ref{sec:security-bound} against observed attack accuracies.

\subsection{Experimental Setup}
\subsubsection*{Dataset}
\label{sec:exp_setup}
% We conduct our experiments using the widely adopted Tor-based website fingerprinting dataset originally collected by \cite{sirinam2018deep}, which has become a standard benchmark in recent WF research, including works such as~\cite{shen2023subverting, shen2024real, zhou2023wf, de2020trafficsliver}. The dataset comprises two parts: a \textit{closed-world} set and an \textit{open-world} set. The closed-world subset includes 1{,}000 traffic traces for each of 95 monitored webpages, totaling 95{,}000 labeled instances. 
% % The open-world subset consists of one trace each from 40{,}716 additional unmonitored webpages, strictly excluding the 95 monitored sites to avoid overlap. 
% All traces were collected via real Tor connections and reflect authentic network conditions including congestion, jitter, and Tor circuit variability. 
% This dataset enables a comprehensive evaluation of both closed-world classification accuracy and open-world true discovery rate under realistic and challenging traffic settings.

 We conduct our experiments on two well-known public WF datasets. The first is collected by Sirinam et al.~\cite{sirinam2018deep}, and it has become a standard benchmark in recent WF research~\cite{shen2023subverting, shen2024real, zhou2023wf, de2020trafficsliver}. For our experiments, we used its \textit{ closed-world} subset, which contains 1,000 traffic traces for each of the 95 monitored websites. The second is the large-scale Automated Website Fingerprinting (AWF) dataset from Rimmer et al.~\cite{rimmer2017automated}. From this collection, we use the subset corresponding to the top 100 most popular websites, which provides 2,500 traffic traces per website. Specifically, both data sets consist of traces collected over live Tor connections, reflecting authentic and variable network conditions, including congestion, jitter, and circuit diversity.

\subsubsection*{Parameter Grid for Adaptive Tamaraw}

Tamaraw is parameterized by three key hyperparameters: the upstream and downstream padding intervals $(\rho_{\text{out}}, \rho_{\text{in}})$, and the length bucket parameter $L$: all traces are padded to multiples of $L$.
%, which dictates the granularity of trace length regularization. 
A higher $L$ increases the likelihood that multiple defended traces will collapse to the same length, reducing fingerprintability. 
We consider three values for $L$: $\{100, 500, 1000\}$. 
%We fix $L$ and search for padding intervals $(\rho_{\text{in}}, \rho_{\text{out}})$ that minimize bandwidth and latency overhead.

%We consider three values for $L$: $\{100, 500, 1000\}$. While $L = 100$ was recommended in the original Tamaraw paper, recent trends motivate the exploration of larger values. Multiple studies have found that modern webpages are increasingly complex, embedding more third-party resources and requiring more data to load~\cite{juarez2014critical,gong2022surakav,singh2024connecting}. This naturally results in longer network traces, broadening the potential space of trace lengths and making a fixed padding threshold of $L = 100$ insufficient. As a result, we also experiment with larger values of $L$.% to better capture the current diversity in trace lengths while balancing overhead and privacy.

% With $L$ fixed, we explore a grid of candidate values for $\rho_{\text{in}}$ and $\rho_{\text{out}}$. Given that incoming traffic typically arrives more frequently than outgoing traffic during webpage loads \cite{cai2014cs, holland2020regulator}, we sweep $\rho_{\text{in}}$ in the range $[0.001, 0.006]$ and $\rho_{\text{out}}$ in the range $[0.005, 0.21]$. For each configuration, we measure the resulting bandwidth and time overhead, which we then use in subsequent analyses to quantify the trade-offs between efficiency and security.

For each $L$, we explore a grid of candidate values for $\rho_{\text{in}}$ and $\rho_{\text{out}}$. Given that incoming traffic typically arrives more frequently than outgoing traffic during webpage loading~\cite{cai2014systematic, holland2020regulator}, we sweep $\rho_{\text{in}}$ in the range $[0.001, 0.006]$ and $\rho_{\text{out}}$ in the range $[0.005, 0.21]$. %For each $(\rho_{\text{in}}, \rho_{\text{out}})$ configuration, we apply the corresponding Tamaraw schedule to the entire dataset and record the resulting bandwidth and time overheads.
We retain only the Pareto-optimal configurations (those for which no other configuration achieves strictly lower bandwidth and time overhead simultaneously \cite{cai2014systematic}). This selection yields 33 unique $(\rho_{\text{in}}, \rho_{\text{out}})$ pairs for the Sirinam et al. dataset and 40 for the AWF dataset, that span the optimal trade-off frontier between efficiency and security. A detailed report of these configurations is provided in Appendix~\ref{appendix:pareto}. 
%These Pareto-optimal parameter sets form the basis for all evaluations and analyses in the subsequent sections. In particular, we select each of these 33 configurations to serve as the global parameter tuple \((\rho_{\text{in}}^g, \rho_{\text{out}}^g)\) in our global-to-local defense strategy, as described in Section~\ref{sec:early_pred}.

\subsubsection*{Model Setup}

As discussed in Section~\ref{sec:early_pred}, our early anonymity set detection relies on a two-stage pipeline: a Holmes model for webpage-level prediction, followed by a lightweight k-fingerprinting (kFP) model for identifying fine-grained traffic patterns. Full details on model architectures are provided in Appendix~\ref{appendix:architectures}.

\subsubsection*{Hyperparameter Configuration}
\label{sec:hyperparams}

For optimal performance and fair evaluation, we performed a grid search over each hyperparameter. Table~\ref{tab:hyperparams} lists the key parameters of Adaptive Tamaraw. The first three ($k$, Max Intra-Webpage Patterns, and $\alpha$) minimize overhead while maintaining security, and the others maximize prediction accuracy for early anonymity set detection.

% For all comparative baselines (Tik-Tok, Random Forest, k-FP, and Laserbeak), we used the original hyperparameters reported in their respective papers.

% Ultra-compact version for narrow columns

\begin{table}[htb]
\centering
\caption{Hyperparameters for Adaptive Tamaraw with grid search spaces.}
\label{tab:hyperparams}
\small
\begin{tabular}{|l|c|c|}
\hline
\textbf{Parameter} & \textbf{Value} & \textbf{Search Space} \\
\hline
\parbox{3.5cm}{\centering $K$ (local scaling)} & 7 & \parbox{3cm}{\centering 1--10} \\
\hline
\parbox{3.8cm}{\centering Max Intra Webpage Patterns} & 6 & \parbox{3cm}{\centering 2--8} \\
\hline
\parbox{3.5cm}{\centering $\alpha$ (ECDIRE)} & 0.9 & \parbox{3cm}{\centering \{0.6,0.7,0.8,0.9, 1\}} \\
\hline
\parbox{3.5cm}{\centering Holmes lr} & 5e-4 & \parbox{3cm}{\centering \{5e-6,5e-5,5e-4,5e-3\}} \\
\hline
\parbox{3.5cm}{\centering TAM Time Slot} & 80ms & \parbox{3cm}{\centering \{40ms, 80ms, 120ms \}} \\
\hline
\parbox{3.5cm}{\centering Holmes batch} & 256 & \parbox{3cm}{\centering \{32,64,128,256\}} \\
\hline
\parbox{3.5cm}{\centering Holmes epochs} & 80 & \parbox{3cm}{\centering fixed} \\
\hline
% \parbox{3.5cm}{\centering kFP trees} & 100 & \parbox{3cm}{\centering fixed} \\
% \hline
\end{tabular}
\end{table}
%This architecture enables timely assignment of incoming traces to their corresponding anonymity sets during the early stages of the page load.

%This two-stage pipeline enables reliable and efficient early-stage classification of traffic traces with minimal computational and memory overhead.

% \subsection{Attacks used for verification}
% \label{subsec:attacks}

% To confirm that empirical attackers respect the theoretical bound, we retrain four representative state-of-the-art classifiers on each defended dataset:

% \begin{description}
%   \item[DF] Deep Fingerprinting convolutional network \cite{sirinam2018df}.
%   \item[Tik-Tok] CNN with direction and timing channels \cite{rahman2019tiktok}.
%   \item[k-FP] Random-forest fingerprinting with feature extraction \cite{hayes2016kfp}.
%   \item[RF] Transformer model operating on Traffic Aggregation Matrices \cite{shen2023rf}.
% \end{description}

% Original hyper-parameters are kept.  
% Training stops when the validation accuracy has not improved for 15 consecutive epochs.

\subsection{In-Training Evaluation}
\label{exp:closed_world}
In the in-training setting, we evaluate the effectiveness of our defense in protecting webpages that were explicitly included in the defense's training dataset.
 Specifically, we conduct our evaluation on the monitored websites from both the Sirinam et al. and AWF datasets. These sites serve as our set of protected targets. Following prior work~\cite{shen2024real, shen2023subverting}, we partition the dataset into training, validation, and testing sets using an 8:1:1 ratio. At a high level, our full defense pipeline proceeds as follows. We begin by applying the modified CAST algorithm \ref{sec:pattern_detection} to the training traces of each webpage to extract recurring traffic patterns. These patterns are then clustered using the $k$-anonymity-based algorithm introduced in Section~\ref{sec:AS_gen} to obtain our anonimity-sets (we perform the experiments by varying the minimum required anonymity set size $k$ from 2 to 30). For each of the obtained anonymity-sets, we identify a safe timestamp using early time-series classification as described in Section~\ref{sec:early_pred}, allowing us to switch from a global Tamaraw configuration to a cluster-specific configuration as early as possible. Finally, we compare the performance against applying the classical Tamaraw to the traces.

For a fixed global configuration $(\rho_{\text{in}}^g, \rho_{\text{out}}^g)$ selected from the Pareto-optimal grid introduced in Section~\ref{sec:exp_setup}, we determine cluster-specific regularization parameters for each anonymity set. Specifically, we sweep over the full grid of $(\rho_{\text{in}}, \rho_{\text{out}})$ pairs and select the configuration that achieves lower bandwidth and time overheads on the traces within the anonymity set, compared to $(\rho_{\text{in}}^g, \rho_{\text{out}}^g)$. This ensures that the local configuration provides a clear improvement over the global baseline for the assigned subset of traces.

We also train the Holmes model on the same training traces to perform early webpage prediction, and independently train a set of kFP classifiers for each webpage to support fine-grained pattern prediction. Using the ECDIRE procedure (described in Section~\ref{sec:early_pred}), we derive a safe prediction time for each $(\text{webpage}, \text{pattern})$ pair. We set $\alpha = 0.9$, meaning a prediction is considered safe once it reaches 90\% of the classification accuracy achieved when the classifier is trained on complete traces. This value was selected empirically: we evaluated a range of $\alpha$ values from 0.5 to 1, and compared the resulting overhead improvements. The experiments showed that $\alpha = 0.9$ yielded the best overhead improvements compared to the original tamaraw.

We proceed to evaluate the performance of Adaptive Tamaraw using the held-out test set. For each test trace $i$, and for each global Pareto regularization pair $(\rho_{\text{in}}^g, \rho_{\text{out}}^g)$, we first apply Tamaraw with the global parameters and record the resulting bandwidth and time overheads. Next, we evaluate the same trace under the complete adaptive Tamaraw strategy: we begin with $(\rho_{\text{in}}^g, \rho_{\text{out}}^g)$ and switch to the specific configuration of the anonymity set $(\rho_{\text{in}}^{\mathcal{S}_{i}}, \rho_{\text{out}}^{\mathcal{S}_{i}})$ once the predicted anonymity set becomes identifiable according to its safe prediction time. We then record the resulting overheads. We repeat this procedure for all the global Pareto regularization pairs $(\rho_{\text{in}}^g, \rho_{\text{out}}^g)$.

Table~\ref{tab:overhead_comparison} presents the average overheads (time and bandwidth, each averaged over all global regularization pairs) for static Tamaraw versus Adaptive Tamaraw at various cluster sizes \(k\), across three bucket lengths \(L \in \{100, 500, 1000\}\) for both datasets.. 
On the Sirinam et al. dataset, the savings are substantial: when $L=1000$, Adaptive Tamaraw reduces the average overheads from 258\% bandwidth and 199\% time overhead to 223\% and 135\% respectively at $k=2$, a combined reduction of 99.0 percentage points. For the AWF dataset, we observe that while the baseline overhead for static Tamaraw is generally higher, the overhead reductions from our defense are still pronounced. For instance, at $L=1000$ and $k=2$ on AWF, the average overheads drop from 182\% and 207\% to 145\% and 154\%, a reduction of 90.0 percentage points.
%Even at $k=30$, the reduction remains substantial across both datasets. 
%Our adaptive strategy is especially effective on datasets with more diverse and higher-overhead traffic patterns, where the potential for optimization is greater.}

The choice of $k$ controls the defense's effectiveness: smaller values result in a looser bound on attacker accuracy (fully explored in Section~\ref{sec:max_accuracy}) while producing the smallest overhead. Interestingly, Adaptive Tamaraw results in higher overhead savings with smaller $k$, i.e. 
{\em when the anonymity sets cover fewer intra-webpage patterns}. 
An intuitive explanation is that in high $k$ settings, each anonymity set is large and therefore does not benefit as much from localized, set-specific adaptive parameters.
More comprehensive results across a broader range of \(k\) values are provided in Appendix~\ref{appendix:overhead-improvement}.

\begin{table*}
\centering
\caption{Average overheads (time and bandwidth) comparison between static Tamaraw and Adaptive Tamaraw across multiple cluster sizes \(k\) and bucket lengths \(L\). Each entry shows the absolute overhead percentage, with the difference from Tamaraw in parentheses. Adaptive Tamaraw consistently reduces overhead, especially for small values of \(k\).}
\label{tab:overhead_comparison}
\small
\resizebox{\textwidth}{!}{%
\begin{tabular}{|l|c|cc|cc|cc|cc|}
\toprule
\multirow{3}{*}{\parbox{2cm}{\centering\textbf{Dataset}}} & 
\multirow{3}{*}{\parbox{1.5cm}{\centering\textbf{L}}} & 
\multicolumn{2}{c|}{\textbf{Tamaraw (Baseline)}} & 
\multicolumn{6}{c|}{\textbf{Adaptive Tamaraw}} \\
\cline{3-10}
& & \multicolumn{2}{c|}{} & \multicolumn{2}{c|}{\textbf{\parbox[c][0.4cm]{1.5cm}{k = 2}}} & \multicolumn{2}{c|}{\textbf{k = 7}} & \multicolumn{2}{c|}{\textbf{k = 30}} \\
\cline{3-10}
& & \textbf{\parbox[c][0.6cm]{0.8cm}{BW}} & \textbf{Time} & \textbf{BW} & \textbf{Time} & \textbf{BW} & \textbf{Time} & \textbf{BW} & \textbf{Time} \\
\midrule
\multirow{3}{*}{\parbox{2cm}{\centering Sirinam et al.}} 
& \parbox[c][0.4cm]{1.5cm}{\centering 100} & 158\% & 83\% & 136\% (\textbf{-22}) & 68\% (\textbf{-15}) & 144\% (\textbf{-14}) & 74\% (\textbf{-9}) & 157\% (\textbf{1}) & 77\% (\textbf{-6}) \\
\cline{2-10}
& \parbox[c][0.4cm]{2cm}{\centering 500} & 198\% & 98\% & 176\% (\textbf{-22}) & 86\% (\textbf{-12}) & 184\% (\textbf{-14}) & 86\% (\textbf{-11}) & 196\% (\textbf{-2}) & 87\% (\textbf{-11}) \\
\cline{2-10}
& \parbox[c][0.4cm]{1.5cm}{\centering 1000} & 258\% & 199\% & 223\% (\textbf{-35}) & 135\% (\textbf{-64}) & 235\% (\textbf{-23}) & 144\% (\textbf{-55}) & 248\% (\textbf{-10}) & 150\% (\textbf{-49}) \\
\midrule
\multirow{3}{*}{\parbox{2cm}{\centering AWF}} 
& \parbox[c][0.4cm]{1.5cm}{\centering 100} & 151\% & 153\% & 100\% (\textbf{-51}) & 111\% (\textbf{-42}) & 109\% (\textbf{-42}) & 123\% (\textbf{-30}) & 125\% (\textbf{-26}) & 135\% (\textbf{-18}) \\
\cline{2-10}
& \parbox[c][0.4cm]{1.5cm}{\centering 500} & 157\% & 183\% & 122\% (\textbf{-35}) & 147\% (\textbf{-36}) & 132\% (\textbf{-25}) & 155\% (\textbf{-28}) & 141\% (\textbf{-16}) & 172\% (\textbf{-6}) \\
\cline{2-10}
& \parbox[c][0.4cm]{1.5cm}{\centering 1000} & 182\% & 207\% & 145\% (\textbf{-37}) & 154\% (\textbf{-53}) & 157\% (\textbf{-25}) & 162\% (\textbf{-45}) & 169\% (\textbf{-13}) & 182\% (\textbf{-25}) \\
\bottomrule
\end{tabular}%
}
\end{table*}

\textbf{Per-Trace Analysis.} Beyond aggregate averages, we examine how Adaptive Tamaraw impacts overhead on a per-trace basis. To this end, we selected a representative global Pareto configuration, $(\rho_{\text{in}} = 0.012, \rho_{\text{out}} = 0.04)$, which was also reported in the original Tamaraw paper~\cite{cai2014systematic}, and offers a balanced trade-off between bandwidth and time overhead. We applied both static Tamaraw and Adaptive Tamaraw (with $k = 7$ and $L = 100$) to each test trace in the Sirinam et al. dataset using this configuration as the global parameters. The resulting average bandwidth and time overheads for static Tamaraw were 106\% and 53\%, respectively. With Adaptive Tamaraw, these dropped to 95\% and 47\%, corresponding to an overall overhead reduction of approximately 17\%.  We measured the difference in total overhead between the two methods for each trace and plotted the distribution.

Figure~\ref{fig:per_trace_savings} shows the resulting distribution of per-trace total overhead savings. Positive values indicate cases where Adaptive Tamaraw yielded lower overhead than static Tamaraw, while negative values represent regressions. Although some traces exhibit higher overhead under adaptation, the overall trend is clearly beneficial: the average saving is positive and approximately 10\% of traces experience savings exceeding 100\%, even reaching up to 500\%. These results demonstrate Adaptive Tamaraw’s ability to reduce overhead, particularly significant for traces that would otherwise be heavily penalized by a one-size-fits-all padding strategy.

% \textbf{Per-Trace Analysis.} Beyond aggregate averages, we examine how Adaptive Tamaraw impacts overhead on a per-trace basis. To this end, we randomly selected two global Pareto configurations and applied both static Tamaraw and Adaptive Tamaraw (with $k = 7$ and $L = 100$) to each test trace. We then measured the difference in total overhead between the two methods.

% Figure~\ref{fig:per_trace_savings} shows the distribution of these savings across all test traces for the two selected configurations. Positive savings values indicate traces that benefited from Adaptive Tamaraw. The red dashed line denotes no savings. We observe that although a number of traces exhibit negative savings, the overall mean is positive in both cases, suggesting a net benefit across the dataset. More notably, a substantial number of traces experience dramatic savings (well above 100\% and in some cases exceeding 500\%). These extreme cases highlight the ability of Adaptive Tamaraw to significantly reduce overhead for the traces that would otherwise suffer the most under the rigid global padding scheme of Tamaraw.

Furthermore, to assess how reliably Adaptive Tamaraw identifies the correct anonymity set at runtime, we evaluated the performance of the combined Holmes + kFP classifier used in the ECDIRE procedure. On average, the correct anonymity set was identified in 81\% of test traces. In 10\% of cases, no set was chosen, resulting in the defense completing using the original Tamaraw parameters. The remaining 9\% of traces were assigned to an incorrect anonymity set. 
%While misclassification can degrade performance, the majority of predictions are accurate, and complete misassignment remain relatively infrequent.

This per-trace analysis underscores a key advantage of adaptation: by tailoring the padding schedule to individual traffic characteristics as soon as sufficient information becomes available, our approach delivers substantial efficiency gains in a fine-grained and targeted manner.

% \begin{figure}[h]
%     \centering
%     \includegraphics[width=0.4\textwidth]{Figures/results/saving_histogram_4.png}
%     \includegraphics[width=0.4\textwidth]{Figures/results/saving_histogram_30.png}
%     \caption{Trace-level overhead savings for two global configurations. Despite some negative cases, a number traces benefit significantly, and the overall mean is positive.}
%     \label{fig:per_trace_savings}
% \end{figure}

\begin{figure}[t]
    \centering
    \includegraphics[width=\linewidth]{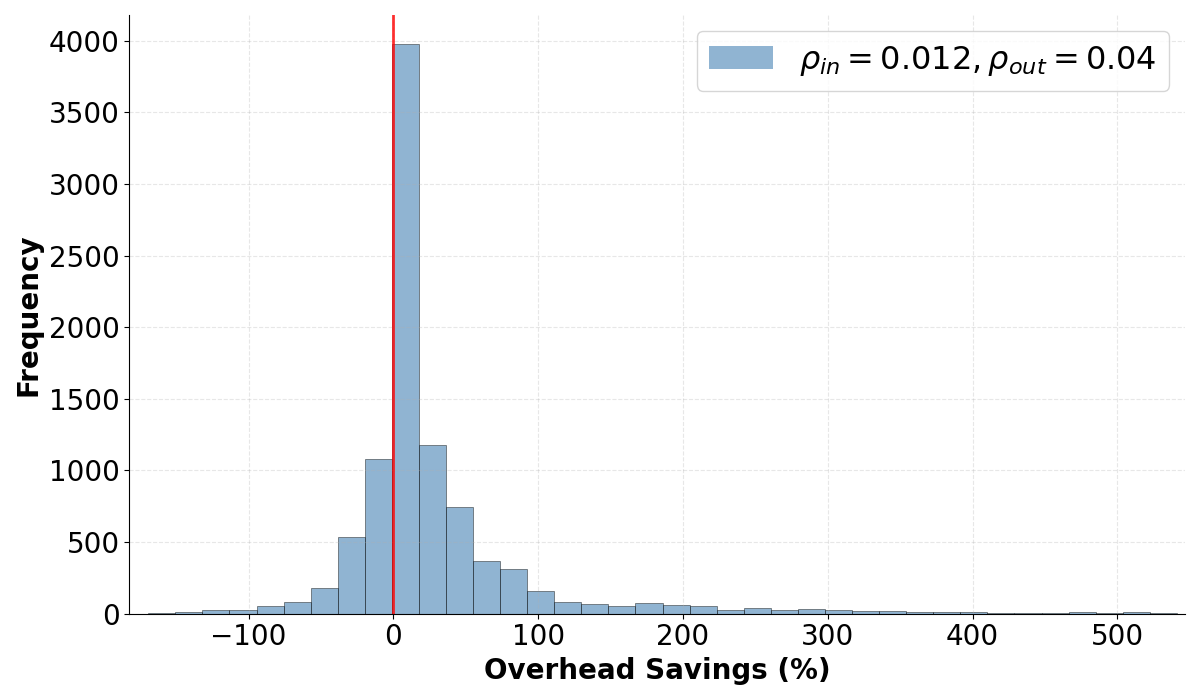}
    \caption{
Distribution of per-trace overhead savings achieved by Adaptive Tamaraw over static Tamaraw for one representative global padding configuration: $(\rho_{\text{in}} = 0.012, \rho_{\text{out}} = 0.04)$,  with $k = 7$ and $L = 100$.The red vertical line at 0\% indicates the point where both methods incur equal overhead; values to the right indicate savings from adaptation, and those to the left indicate higher cost.
    \label{fig:per_trace_savings}
}
\end{figure}

\subsection{Out-of-Training Evaluation}
\label{exp:open_world}
In the out-of-training setting, our goal is to assess how well Adaptive Tamaraw generalizes to website traffic that was not included during the training phase. This scenario is particularly important, as prior anonymity set–based defenses~\cite{nithyanand2014glove, wang2014effective, shen2024real} are typically constrained to protect only \emph{in-training} webpages, limiting their applicability in real-world browsing environments where users may visit previously unseen pages.

We split the Sirinam et al. closed-world dataset, which comprises 95 webpages, into two disjoint subsets of webpages. One subset (containing approximately half of the webpages) serves as the training set for the defense, while the other subset is reserved for the out-of-training evaluation, simulating previously unseen webpages.

%This setup ensures a fair comparison: the defense still has access to sufficient training data to construct reliable anonymity sets, while the open-world traces remain challenging and diverse.

For this experiment, we instantiate Adaptive Tamaraw with $k = 7$ and $L = 100$, a configuration that strikes a practical balance between efficiency and protection. The choice of $L = 100$, inspired by prior work~\cite{cai2014systematic}, provides reasonable security guarantees while avoiding the excessive overhead associated with larger values. Meanwhile, $k = 7$ offers an effective trade-off: it yields notable reductions in bandwidth and time overhead (Table~\ref{tab:overhead_comparison}) and maintains low prediction accuracy within anonymity sets (Section~\ref{sec:max_accuracy}) (i.e., under this setting, the maximum theoretical success rate of any attacker remains below 50\%).
% meeting the threshold commonly associated with resistance to adversarial training~\cite{shen2024real}.

We follow the same training procedure as in the in-training setting, including pattern extraction, anonymity set generation, classifier training, and safe time  on the first half of webpages.
During evaluation of the testing traces, the defense initially applies the global Tamaraw parameters from Section~\ref{sec:exp_setup} and then attempts to assign the trace to one of the anonymity sets, constructed exclusively from the training webpages, at any of the precomputed safe times. If such an assignment is possible, the defense switches to the corresponding anonymity set–specific parameters for the remainder of the trace.

% To compare the efficiency of different defense strategies under practical constraints, we fix a series of time overhead thresholds and, for each threshold, identify  
% the configuration that yields the lowest bandwidth overhead while keeping the time overhead below the given limit. We report this value for three scenarios: (i) the original Tamaraw defense, (ii) Adaptive Tamaraw on out-of-training webpages, and (iii) Adaptive Tamaraw on in-training webpages.
% We threshold time overhead as it represents the maximum tolerance a user can be expected to accept for a defense: time delays directly affect user experience while bandwidth affects network load.

To compare the efficiency of different defense strategies under practical user experience constraints, we evaluate their performance across a series of fixed \textbf{time overhead constraints}. For each constraint (e.g., requiring time overhead to be less than 10\%), we identify the defense configuration from our Pareto-optimal set that yields the lowest possible bandwidth overhead while respecting that time limit. This methodology allows us to answer a practical question: Given a maximum acceptable delay, what is the minimum bandwidth cost? We report these values for three scenarios: (i) the original Tamaraw defense, (ii) Adaptive Tamaraw on out-of-training webpages, and (iii) Adaptive Tamaraw on in-training webpages.

Table~\ref{tab:ow-bandwidth-time} presents the results of this analysis, where each column represents a different time overhead constraint. The values in the table show the minimum bandwidth overhead achieved by the best configuration that satisfies the given time constraint. Adaptive Tamaraw on in-training webpages shows substantial gains over the original Tamaraw, especially under tight time constraints. For instance, when the time overhead is limited to 10\%, Adaptive Tamaraw (in-training) achieves a reduction of 52 percentage points in bandwidth overhead compared to the baseline. Notably, these improvements are more pronounced at lower time overhead thresholds, which is particularly practical: minimizing additional delay is often crucial for preserving user experience.

\begin{table*}[htbp]
\centering
\caption{Bandwidth overhead for Tamaraw and Adaptive Tamaraw under fixed time overhead constraints in both in-training and out-of-training settings. Parentheses show percentage reductions relative to Tamaraw. Adaptive Tamaraw consistently reduces minimum and average bandwidth overhead, with largest reductions in-training.}

\label{tab:ow-bandwidth-time}

\small
\resizebox{0.68\textwidth}{!}{%
\begin{tabular}{l|cccc|c}
\toprule
\multirow{2}{*}{\parbox{2.5cm}{\centering\textbf{Fixed Time Overhead}}} & 
\multicolumn{4}{c|}{\textbf{Minimum Bandwidth Overhead (\%)}} & 
\textbf{Average Bandwidth} \\
& \textbf{$<$ 10\%} & \textbf{$<$ 45\%} & \textbf{$<$ 125\%} & \textbf{$<$ 250\%} & \textbf{Overhead (AUC)} \\
\midrule
\parbox{2.5cm}{\centering Tamaraw} & 279 & 124 & 83 & 66 & 82 \\
\midrule
\parbox{2.5cm}{\centering Adaptive Tamaraw (out-of-training)} & 277 (\textbf{-2}) & 119 (\textbf{-5}) & 80 (\textbf{-3}) & 64 (\textbf{-2}) & 79 (\textbf{-3} )\\
\midrule
\parbox{2.5cm}{\centering Adaptive Tamaraw (in-training)} & 227 (\textbf{-52}) & 110 (\textbf{-14}) & 72 (\textbf{-11}) & 61 (\textbf{-5}) & 68 (\textbf{-14}) \\
\bottomrule
\end{tabular}%
}
\end{table*}

% Table~\ref{tab:ow-bandwidth-time} presents the minimum bandwidth overhead achieved by each defense method under varying fixed time overhead thresholds. Adaptive Tamaraw on in-training webpages shows substantial gains over the original Tamaraw, especially under tight time constraints. For instance, when the time overhead is limited to 10\%, Adaptive Tamaraw (in-training) achieves a reduction of 52 percentage points in bandwidth overhead compared to the baseline. Notably, these improvements are more pronounced at lower time overhead thresholds, which is particularly practical: minimizing additional delay is often crucial for preserving user experience.
%The adaptive strategy excels in this regime by optimizing bandwidth usage while strictly adhering to tight timing constraints.

In the out-of-training setting, where all test traces are from previously unseen webpages, Adaptive Tamaraw still manages to slightly outperform the global Tamaraw strategy. The reductions are modest (e.g., 2--7 percentage points), but consistent across all thresholds. This suggests that even for webpages not present during training, partial trace similarity to known anonymity sets can still yield meaningful gains.

Along with minimum bandwidth values, we report the average bandwidth overhead, computed as the normalized area under the Pareto curve. Both in-training and out-of-training versions of Adaptive Tamaraw show lower average overhead than the original Tamaraw. When Adaptive Tamaraw is trained on webpages the user would visit, it achieves its best performance, but it still outperforms Tamaraw even if it is not.

\subsection{Attacker Accuracy}
\label{sec:max_accuracy}

\subsubsection{Theoretical Bounds}

In this section, we provide a numerical evaluation of the average success probability bound stated in Theorem~\ref{thm:delta-global}, which measures the level of security offered by Adaptive Tamaraw. We use the anonymity sets generated from the in-training evaluation (Section~\ref{exp:closed_world}).

For each fixed combination of anonymity set size $k$ and Tamaraw's bucket length parameter $L$, we consider the corresponding anonymity sets and evaluate the success probability bound for all 33 Pareto-optimal $(\rho_{\text{in}}, \rho_{\text{out}})$ pairs identified in Section~\ref{sec:exp_setup}. For each pair, we compute the right-hand side of the inequality in Theorem~\ref{thm:delta-global-maintext}. 
We then average these bounds across all 33 configurations to obtain a success bound for each $(k, L)$ pair.

% The final results, presented in Figure~\ref{fig:max_acc_vs_k_l}, demonstrate how the success probability of an optimal passive adversary decreases as the anonymity set size increases, and how the level of security varies under different choices of $L$.

Figure~\ref{fig:max_acc_vs_k_l} presents the computed theoretical upper bounds on the attack success probability, as described in Theorem~\ref{thm:delta-global}, across varying anonymity set sizes $k$ and for three different values of the bucket length parameter $L \in \{100, 500, 1000\}$. As expected, the attacker's success probability decreases with larger anonymity set sizes, confirming that increasing $k$ heightens adversarial uncertainty and weakens passive WF attacks. Interestingly, the drop is not inversely proportional to $k$, indicating that larger $k$ leads to sets containing many traces from the same site. Increasing $L$ (padding length) also lowers maximum attacker accuracy, for example, at $k=7$, success falls from about 45\% with $L=100$ to 35\% with $L=1000$, though with substantial overhead, as shown in Table~\ref{tab:overhead_comparison}.

%The results demonstrate the tightness and practicality of the bound in Theorem~6.2, highlighting how both structural parameters of the defense, cluster granularity via $k$ and padding strength via $L$, jointly determine the defense’s theoretical robustness. This analysis provides a formal justification for the empirical trends observed in earlier sections.

\begin{figure}[t]
    \centering
    % The graphic below assumes the file was exported to PDF or PNG and placed
    % in your figures/ directory.  Adjust the path or \includegraphics options
    % as needed.
    \includegraphics[width=\linewidth]{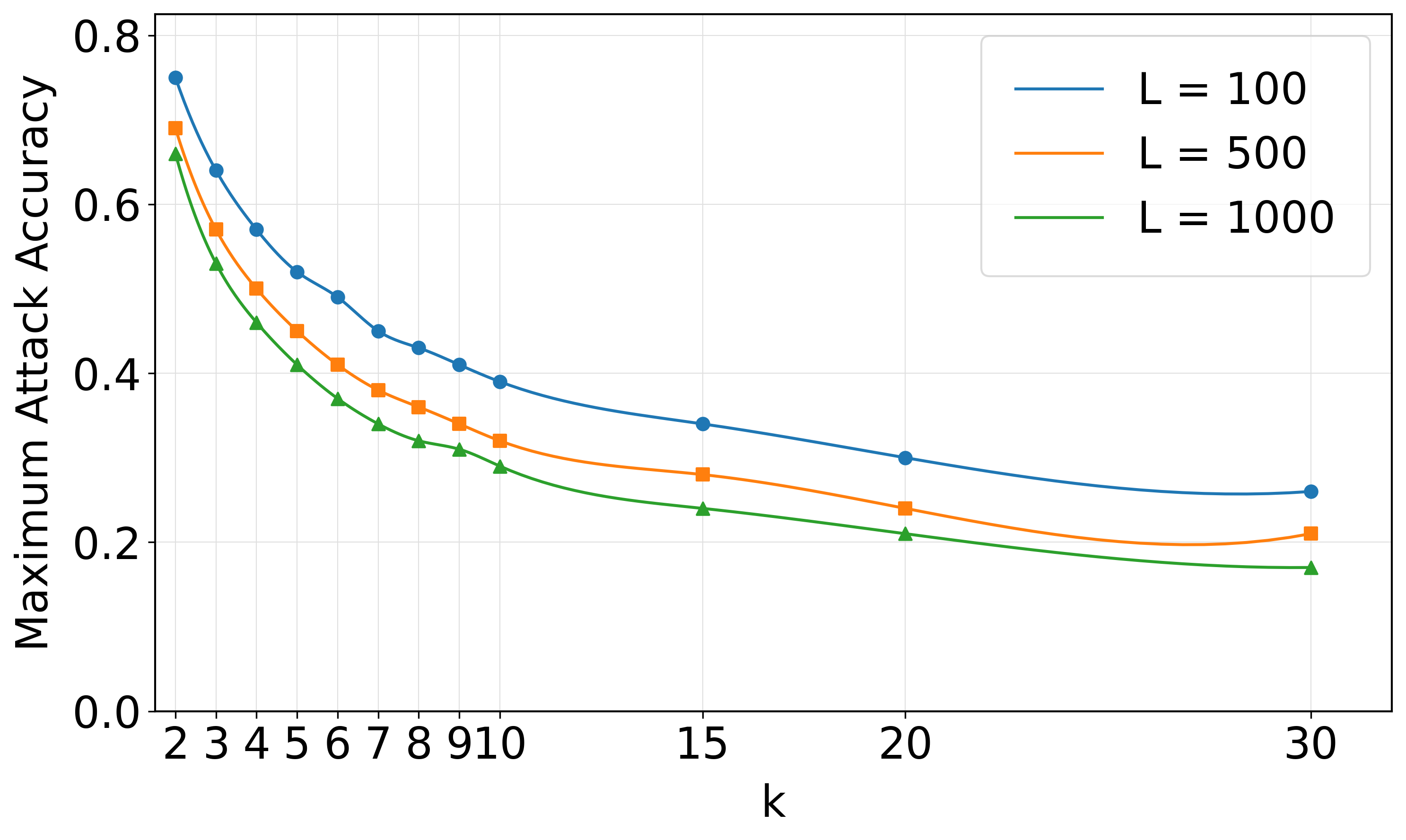}
    \caption{Formal upper bound on attacker accuracy for different anonymity-set sizes~$k$ and bucket-length parameters~$L$ (Section \ref{sec:security-bound}). Increasing $k$ enlarges each trace's pre-image, while larger~$L$ suppresses size-based leakage. Together they tighten the bound from $\approx\!0.75$ at $(k,L)=(2,100)$ to below~$0.17$ at $(k,L)=(30,1000)$, illustrating the security-overhead trade-off.}
    \label{fig:max_acc_vs_k_l}
\end{figure}

We also evaluated the theoretical bound in the \textit{out-of-training} scenario with $k = 7$ and $L = 100$. We assigned each trace to its predicted anonymity set, computed the corresponding bound, and averaged across all traces. This yielded a maximum attacker accuracy of 31\%, confirming the theoretical framework remains valid for unseen webpages. This bound is lower than the \textit{in-training} counterpart in Figure~\ref{fig:max_acc_vs_k_l} (45\%), possibly because the defense reverts toward original Tamaraw behavior out-of-training, as evidenced by modest bandwidth improvements in Table~\ref{tab:ow-bandwidth-time}. This produces more uniform parameters and homogeneous traffic patterns, lowering attack success rates.

\subsubsection{Empirical Attacks}

To validate the theoretical success probability bounds derived in the previous section, we evaluate the effectiveness of state-of-the-art WF attacks on traces defended by Adaptive Tamaraw, and compare against the corresponding theoretical upper bounds to assess the tightness of the bound presented in Theorem~6.2. We consider four widely-used and high-performing WF attack models: \textbf{kFP}~\cite{hayes2016k}, \textbf{Tik-Tok} ~\cite{rahman2019tik}, \textbf{RF} ~\cite{shen2023subverting}, and \textbf{LASERBEAK}~\cite{mathews2024laserbeak}.

% We consider four widely-used and high-performing WF attack models: (1) \textbf{kFP}, a random forest–based fingerprinting method~\cite{hayes2016k}; (2) \textbf{Tik-Tok}, a deep fingerprinting convolutional network~\cite{rahman2019tik}; (3) \textbf{RF}, a deep fingerprinting model using a combination of 2D and 1D convolutional layers trained on Traffic Aggregation Matrices~\cite{shen2023subverting}; and (4) \textbf{LASERBEAK}, a recent convolutional vision transformer architecture~\cite{wu2021cvt}, operating on multi-channel feature representations~\cite{mathews2024laserbeak}.

In alignment with Section~\ref{exp:open_world}, we use the practical configuration of $k = 7$ and $L = 100$. We randomly sample four $(\rho_{\text{in}}, \rho_{\text{out}})$ pairs from the 33 Pareto-optimal global configurations and apply Adaptive Tamaraw to the training traces using each selected pair. %For each configuration, we compute the theoretical success bound as described previously.
We then conduct adversarial training using the defended traces and evaluate the attack success rate of each WF model on the corresponding test set. The training protocols and hyperparameters for all four models are the same mentioned in their original paper.%The comparison of theoretical bounds and empirical accuracies allows us to assess both the robustness of the defense and the practical sharpness of our theoretical analysis.

Table~\ref{tab:empirical-vs-theoretical} compares the empirical classification accuracy of the four WF attack models against the theoretical upper bounds. Across all evaluated cases, we observe that the empirical success rate of each model remains consistently below the theoretical bound, where RF is the model achieving the closest results to the bound.
 This shows that the bound provided in Theorem~6.2 is useful as a principled estimate of attack performance.

%Among the evaluated models, the RF transformer approaches the theoretical limit more closely than other models. This suggests that RF may be more effective at extracting residual traffic features that remain even after defense, possibly due to its hierarchical convolutional architecture designed to capture both temporal and directional patterns in structured time-series inputs.

%Overall, the results validate both the effectiveness of Adaptive Tamaraw and the informativeness of the theoretical bounds. They show that while strong attacks like RF are capable of partial inference, the defense is successful in limiting all adversaries to accuracy levels well below the theoretical threshold, providing provable and measurable privacy benefits.

\begin{table}[t]
    \centering
    \caption{Theoretical maximum attacker accuracy bound versus empirical accuracy of website fingerprinting attacks for four randomly selected $(\rho_{\text{out}}, \rho_{\text{in}})$ padding rate pairs with $k = 7$ and $L = 100$.}
    %Empirical accuracies remain consistently below the theoretical bound, validating the security guarantees of the defense under different padding configurations.}

    \label{tab:empirical-vs-theoretical}
    \begin{tabular}{@{}cccccccc@{}}
        \toprule
        & \multicolumn{2}{c}{Padding rates} & \multicolumn{5}{c}{Accuracy (\%)} \\
        \cmidrule(lr){2-3}\cmidrule(l){4-8}
        & $\rho_{\text{out}}$ & $\rho_{\text{in}}$ & Bound & kFP & Tik-Tok & RF & LASERBEAK \\ \midrule
        & 0.012 & 0.030 & 41 & 31 & 20 & 39 & 31 \\
        & 0.009 & 0.010 & 43 & 33 & 21 & 38 & 32 \\
        & 0.010 & 0.020 & 42 & 31 & 22 & 40 & 30 \\
        & 0.030 & 0.120 & 35 & 20 & 15 & 33 & 23 \\ \bottomrule
    \end{tabular}
\end{table}

% \subsection{Obtaining Parameters}

%% file: Sections_revision/8_conclusion.tex
\section{Discussion and Conclusion}
\label{sec:discussion}

We proposed a defense framework combining regularization and supersequence methods with three components: pattern-level clustering, $(k,l)$-diverse anonymity sets, and early-time switching. Instantiated as \textit{Adaptive Tamaraw}, it maintains provable guarantees while adapting to trace structure. Based on our experiments, adaptive Tamaraw offers meaningful gains when tested on webpages it was trained on, and performs similarly to the original Tamaraw on webpages it was not trained on. As $k$ increases, the defense approaches classic Tamaraw, enabling control over the privacy--overhead trade-off (Table~\ref{tab:overhead_comparison}, Figure~\ref{fig:max_acc_vs_k_l}).

While we instantiate our approach using Tamaraw due to its formal security bound,  our approach is general and can be extended to other defenses that support fixed regularization. Each stage of our approach presents opportunities for future research. For instance, more advanced representation learning could enhance pattern detection, while novel clustering algorithms could further optimize anonymity set generation. That said, the computational analysis in Appendix~\ref{appendix:costs} shows that our current instantiation of Adaptive Tamaraw adds negligible latency and memory overhead on standard hardware. This makes it suitable for practical use in the Tor ecosystem, such as a Pluggable Transport in the WFDefProxy framework~\cite{gong2023wfdefproxy}. With inference latency under 2~ms, it can classify and adjust padding parameters in real time without computational bottlenecks, making it a viable option for Tor Browser integration.

%% file: Sections_revision/appendix.tex
% appendix.tex
% Appendix for "Lightening the Load: A Cluster-Based Framework for Lower-Overhead, Provable Website Fingerprinting Defense"

\appendices

\section{Model Architectures}
\label{appendix:architectures}

\textbf{Holmes.} We use the Holmes architecture from Deng et al.~\cite{deng2024robust}, a hybrid convolutional encoder with 2 blocks of Conv2D layers, followed by 4 blocks of Conv1D layers, trained with supervised contrastive learning to align partial traffic traces with full-trace embeddings for early-stage website fingerprinting.

\textbf{k-Fingerprinting (kFP).} We use kFP~\cite{hayes2016k}, a random forest with 100 trees per (website, safe time) pair, based on packet and timing statistics; fingerprints are leaf ID sequences classified by $k$-NN.

% Option 2: Two separate tables (if

\section{Additional Experimental Results}
\label{appendix:addexperiments}
\subsection{Tamaraw Parameter Grid and Pareto Frontier}
\label{appendix:pareto}

To characterize the trade-off between bandwidth and time overhead in the Tamaraw defense, we constructed a grid of $(\rho_{\text{in}}, \rho_{\text{out}})$ pairs following the method recommended in the original Tamaraw paper~\cite{cai2014cs}. Specifically, we started with the original values used in that work: $\rho_{\text{in}}^{\text{init}} = 0.012$ and $\rho_{\text{out}}^{\text{init}} = 0.04$, and varied each parameter independently across 14 values, ranging from 7 times smaller to 7 times larger on a logarithmic scale. That is, each consecutive value differs from the previous by a multiplicative factor of \(e^{\ln(7)/7} \approx 1.316\). This results in a total of \(14 \times 14 = 196\) combinations of $(\rho_{\text{in}}, \rho_{\text{out}})$ pairs.

For each configuration, we applied Tamaraw  and measured the overheads on the traces. We kept only Pareto-optimal pairs, those where no other configuration achieves lower overhead in both metrics. Figure~\ref{fig:pareto_L_100} shows the 33 points fpr Sirinam et al. frontiers for \(L = 100\). Each red dot marks a Pareto-optimal \((\rho_{\text{in}}, \rho_{\text{out}})\) with its (bandwidth, time) overhead. 

% The curve highlights the trade-off: lowering one metric generally increases the other.

% For each configuration, we applied the Tamaraw defense and measured both the bandwidth and time overhead using a batch of traces, fixing the parameter \(L = 100\). Among the resulting overhead pairs, we retained only the Pareto-optimal configurations (those for which no other parameter pair simultaneously achieves lower bandwidth and time overhead). The 33 obtained points in the Sirinam et al. dataset and 40 for AWF  form the Pareto frontier, representing the optimal trade-off curve between efficiency and latency.

% Figure~\ref{fig:pareto_L_100} illustrates the resulting Pareto frontiers for the Sirinam et al. dataset. Each red dot corresponds to a Pareto-optimal $(\rho_{\text{in}}, \rho_{\text{out}})$ configuration, annotated with its corresponding (bandwidth, time) overhead pair. The curve reveals the inherent trade-off: significant reductions in one metric typically come at the cost of increases in the other.

\begin{figure}[htb]
  \centering
  \includegraphics[width=0.95\linewidth]{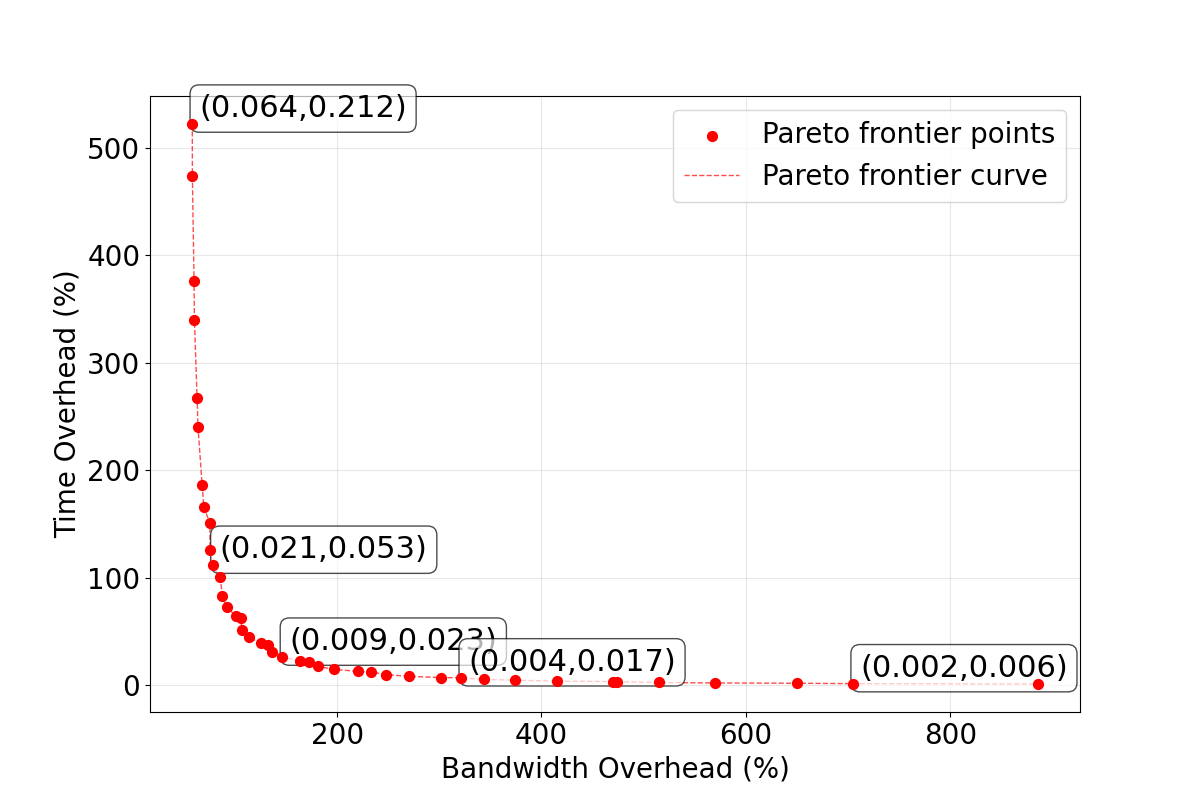}
  \caption{Pareto frontier for Tamaraw overheads computed from 196 parameter pairs for \(L = 100\). Each point corresponds to a $(\rho_{\text{in}}, \rho_{\text{out}})$ configuration that is Pareto-optimal with respect to bandwidth and time overhead.}
  \label{fig:pareto_L_100}
\end{figure}

\subsection{Security Bound Analysis}
\label{appendix:sec_bound}

Our security bound in Section \ref{sec:security-bound} is designed to bound the attacker's average success rate, assuming the attacker fully knows the trace and the underlying trace information, including the switching time $\tau_s$, the chosen anonymity set and the sequence length. 
Here we evaluate how much information is individually leaked by a $\tau_s$-aware adversary who, based on the observed switching time and post-switch rate pair, can infer the corresponding anonymity set.  We performed this evaluation across all 33 Pareto-optimal global configurations (with $L = 100$ and $k = 7$), comparing the empirical success rate of this adversary against the theoretical security bound derived in Section~\ref{sec:security-bound}. All defense parameters, including anonymity sets and safe timestamps, are derived from the training set, and the adversary's performance is measured on the held-out test set.

Our analysis reveals that the switching mechanism does indeed leak the anonymity set identity with high probability: on average, the correct anonymity set accounts for 95.64\% of the traces sharing a specific $(\tau_s, \text{rates})$ tuple. Our theoretical security bound has already assumed such leakage. As shown in Figure~\ref{fig:tau-aware-attacker}, the $\tau$-aware attacker's empirical accuracy consistently remains below the theoretical security bound across all configurations. This confirms that even when the adversary successfully identifies the anonymity set via the switching signal, the $k$-anonymity and $l$-diversity properties within that set remain sufficient to cap their success rate. Furthermore, a gap exists between the theoretical bound and the attacker's actual performance, which might also suggest that the presence of misclassified and non-switching traces introduces additional ambiguity, further confusing the adversary. 
% The histogram in the figure quantifies this ambiguity by plotting the distribution of the number of traces that share the same switching parameters.

\begin{figure}[htb]
  \centering
  \includegraphics[width=0.95\linewidth]{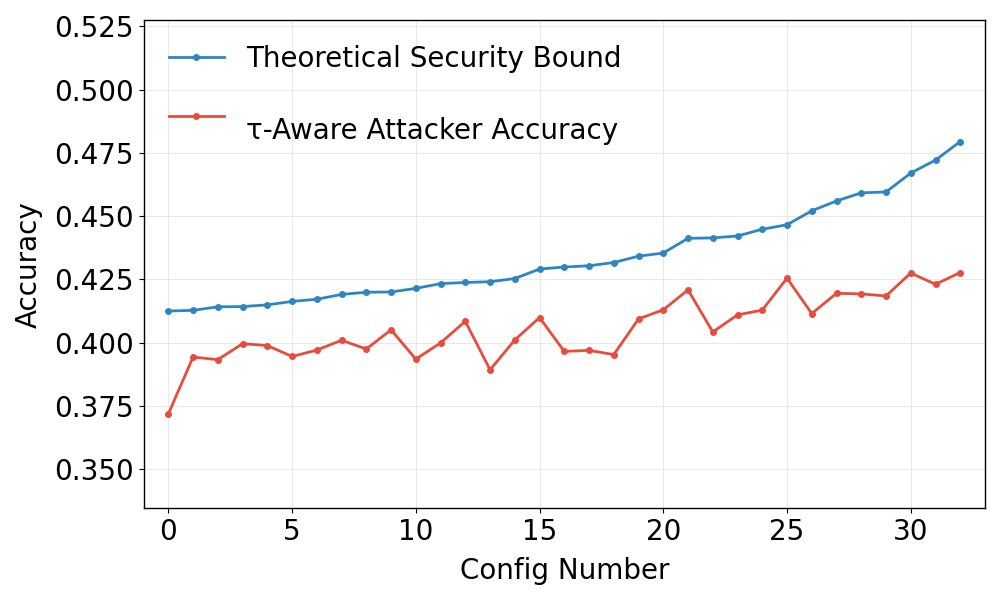}
  \caption{Comparison of the theoretical security bound and the empirical accuracy of a $\tau$-aware attacker across 33 global configurations ($L = 100, k = 7$). The attacker consistently performs below the bound, demonstrating that the defense remains robust even when the switching signal reveals the anonymity set.
}
  \label{fig:tau-aware-attacker}
\end{figure}

\subsection{Generalization to Onion Service Traffic}
 To further evaluate the robustness of our method in out-of-training conditions, we conduct an experiment using traces from onion services, which can exhibit different traffic characteristics than standard websites. We use the onion service dataset collected by Overdorf et al.~\cite{overdorf2017unique}, which contains traffic from distinct services, from which we use 100 sites. The experimental setup is the same as our primary out-of-training evaluation: we construct anonymity sets using only the Sirinam et al. dataset, and then apply both the static and adaptive defenses to the unseen onion traces to assess our method's ability to generalize to this entirely different domain. Table~\ref{tab:ow-bandwidth-time-onion} compares the minimum bandwidth overhead between Adaptive Tamaraw and static Tamaraw under fixed time overhead constraints. For ease of comparison, we have also included the results for normal out-of-training websites from Section VI.C. The data shows that Adaptive Tamaraw provides improvements in both scenarios. For normal out-of-training websites, the adaptive approach yields modest but consistent reductions in bandwidth overhead of up to 5 percentage points. Similarly, Adaptive Tamaraw reduces bandwidth overhead across all constraints on the onion service traffic. Under a 20\% time overhead limit, our approach reduces bandwidth overhead by 8 percentage points for onion websites, demonstrating that even with fundamentally different website traffic, our adaptive mechanism identifies structural similarities to provide improvements.

\begin{table}[htbp]
\centering
\caption{Comparison of minimum bandwidth overhead for Tamaraw and Adaptive Tamaraw under various fixed time overhead constraints when applied on out-of-training websites (both normal and onion websites). Values in parentheses represent percentage point reductions relative to Tamaraw.}
\label{tab:ow-bandwidth-time-onion}
\small
\begin{tabular}{l|l|cccc}
\toprule
\parbox{1cm}{\textbf{Website}} & \parbox{1.2cm}{\textbf{Method}} & \multicolumn{4}{c}{\parbox{4cm}{\centering\textbf{Minimum Bandwidth Overhead (\%)}}} \\
& & \parbox{0.7cm}{\centering\textbf{$<$20\%}} & \parbox{0.7cm}{\centering\textbf{$<$45\%}} & \parbox{0.8cm}{\centering\textbf{$<$125\%}} & \parbox{0.8cm}{\centering\textbf{$<$200\%}} \\
\midrule
\multirow{2}{0.8cm}{Normal} & Tamaraw & 168 & 124 & 83 & 67 \\
& Adaptive & 166 (\textbf{-2}) & 119 (\textbf{-5}) & 80 (\textbf{-3}) & 67 (\textbf{0}) \\
\midrule
\multirow{2}{0.8cm}{Onion} & Tamaraw & 195 & 120 & 78 & 66 \\
& Adaptive & 187 (\textbf{-8}) & 116 (\textbf{-4}) & 77 (\textbf{-1}) & 65 (\textbf{-1}) \\
\bottomrule
\end{tabular}
\end{table}

\subsection{Overhead Improvement Across Cluster Sizes}
\label{appendix:overhead-improvement}
 We present the overhead reduction achieved by Adaptive Tamaraw compared to Tamaraw depending on cluster sizes. For each test trace in the Sirinam et al. dataset, we compute the average improvement for bandwidth and time overhead individually and average the results. Figure~\ref{fig:overhead_vs_k_L} illustrates the average overhead improvement achieved by Adaptive Tamaraw over static global Tamaraw configurations, plotted as a function of the cluster size $k$. Across all values of $L$, we observe that Adaptive Tamaraw consistently reduces overhead compared to the baseline. The largest reductions (observed for $L = 1000$) are seen in the time overhead, which is cut by over 50 percentage points for small $k$. For smaller bucket lengths ($L = 100$ and $L = 500$), the savings are more modest but are more evenly distributed between both bandwidth and time improvements. In all scenarios, smaller anonymity sets (i.e., decreasing $k$) generally lead to greater improvements, and furthermore when $L$ is large.

% This trend highlights the benefit of clustering traffic patterns into larger and more diverse anonymity sets, which allows the second phase of the defense to apply more relaxed parameters while preserving provable security. It also confirms that the adaptive switching strategy is particularly effective when the regularization budget (controlled by $L$) is high—leading to large opportunities for savings through early switching. Interestingly, the improvements are often most pronounced when $k$ is small (reaching around 100\% reduction when $L = 1000$, and 40\% when $L = 100$). This may be due to the fact that smaller anonymity sets contain more homogeneous traffic patterns, allowing the anonymity set–specific Tamaraw parameters to better optimize a narrower distribution of traces. As $k$ increases, the diversity within each set grows, which can dilute the effectiveness of a single shared configuration across that set.
% \begin{figure}[htb]
%   \centering
%   \includegraphics[width=0.95\linewidth]{Figures/results/combined_overhead_improvements.png}
%   \caption{Average total overhead improvement of Adaptive Tamaraw compared to global Tamaraw, plotted across different anonymity set sizes $k$ for three values of $L$.}
%   \label{fig:overhead_vs_k_L}
% \end{figure}

\begin{figure}[htb]
  \centering
  \includegraphics[width=0.95\linewidth]{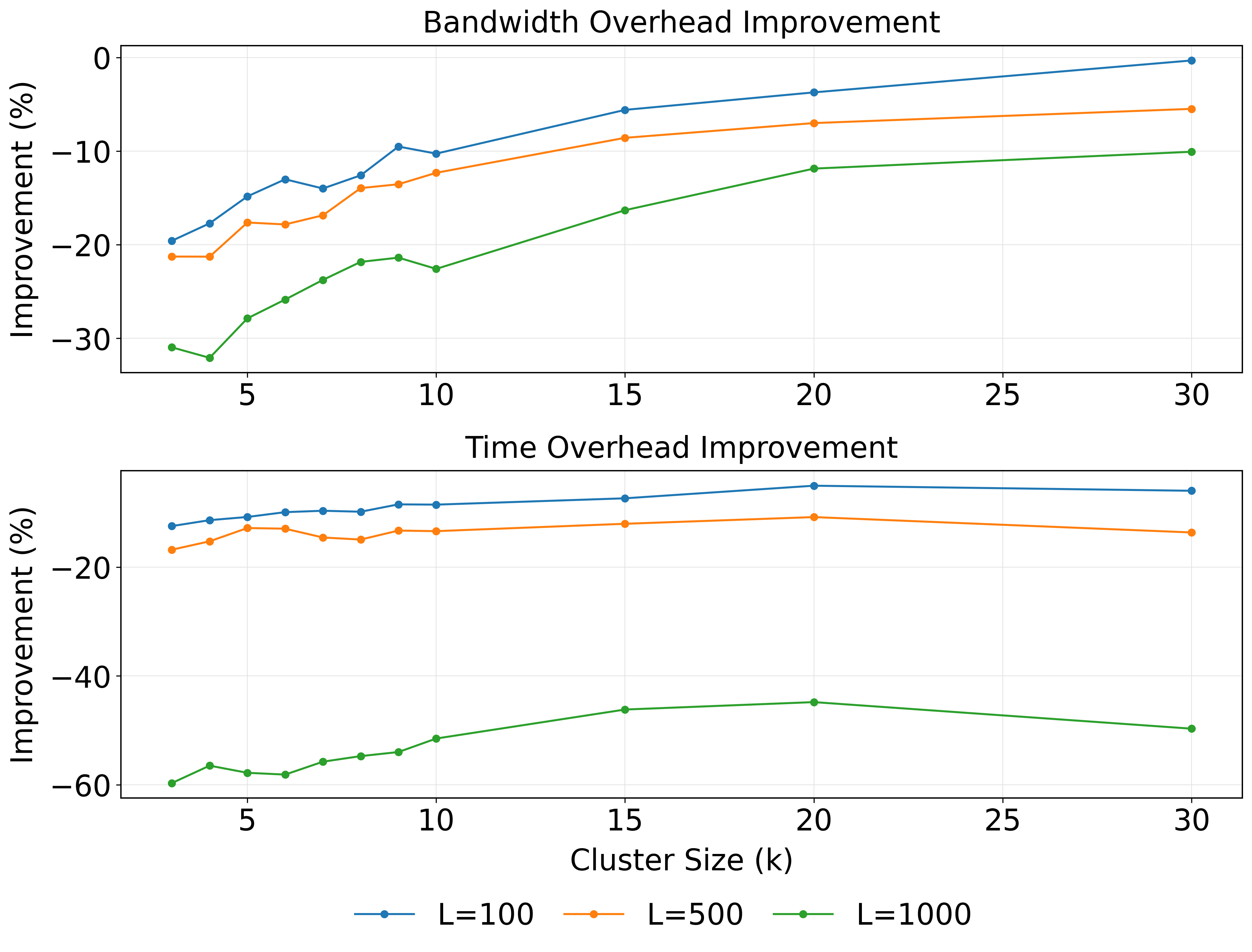}
  \caption{Average total overhead improvement of Adaptive Tamaraw compared to global Tamaraw, plotted across different anonymity set sizes $k$ for three values of $L$.}
  \label{fig:overhead_vs_k_L}
\end{figure}

\section{Ablation Studies}
\subsection{Cluster Diversity Evaluation}
\label{appendix:cluster_diversity}

As discussed in Section~\ref{sec:AS_gen}, our $k$-anonymity clustering algorithm uses a Tamaraw-specific distance function that encourages low intra-cluster diversity. To validate this, we evaluate the purity of generated anonymity sets.

Following~\cite{singhal2005clustering}, we compute cluster $i$'s purity $p_i$ as $\frac{\max_j N_{i,j}}{N_{p_i}}$, where $N_{i,j}$ is the number of traces from website $j$ in cluster $i$, and $N_{p_i}$ is the total traces in cluster $i$. Lower purity indicates greater diversity. We compute average purity across all clusters for each $k$ with fixed $L = 100$. Figure~\ref{fig:average_purities} shows these values alongside the theoretical lower bound $f(k) = 1/k$, representing optimal diversity for a perfectly balanced cluster. As shown in the figure, the empirical purity values closely track the $1/k$ curve, indicating that our clustering procedure approaches the theoretical optimum in terms of diversity. 
% As discussed in Section~\ref{sec:AS_gen}, our proposed $k$-anonymity-based clustering algorithm relies on a Tamaraw-specific distance function that implicitly encourages low intra-cluster diversity. To empirically validate this behavior, we evaluate the purity of the anonymity sets generated by our clustering algorithm.

% Following standard definitions~\cite{singhal2005clustering}, we compute the purity $p_i$ of the $i$-th cluster as:
% \[
% p_i \triangleq \frac{\max_j N_{i,j}}{N_{p_i}} \times 100\%,
% \]
% where $N_{i,j}$ is the number of traces in cluster $i$ that originate from website $j$, and $N_{p_i}$ is the total number of traces in cluster $i$. The purity thus reflects the dominance of the most frequent website in a given cluster, with lower values indicating greater diversity.

% We compute the average purity across all clusters for each anonymity set size $k$ with fixed $L = 100$. Figure~\ref{fig:average_purities} presents the resulting values alongside the theoretical lower bound $f(k) = 1/k$, which represents the best achievable diversity for a perfectly balanced cluster.

\begin{figure}[htb]
    \centering
    \includegraphics[width=0.45\textwidth]{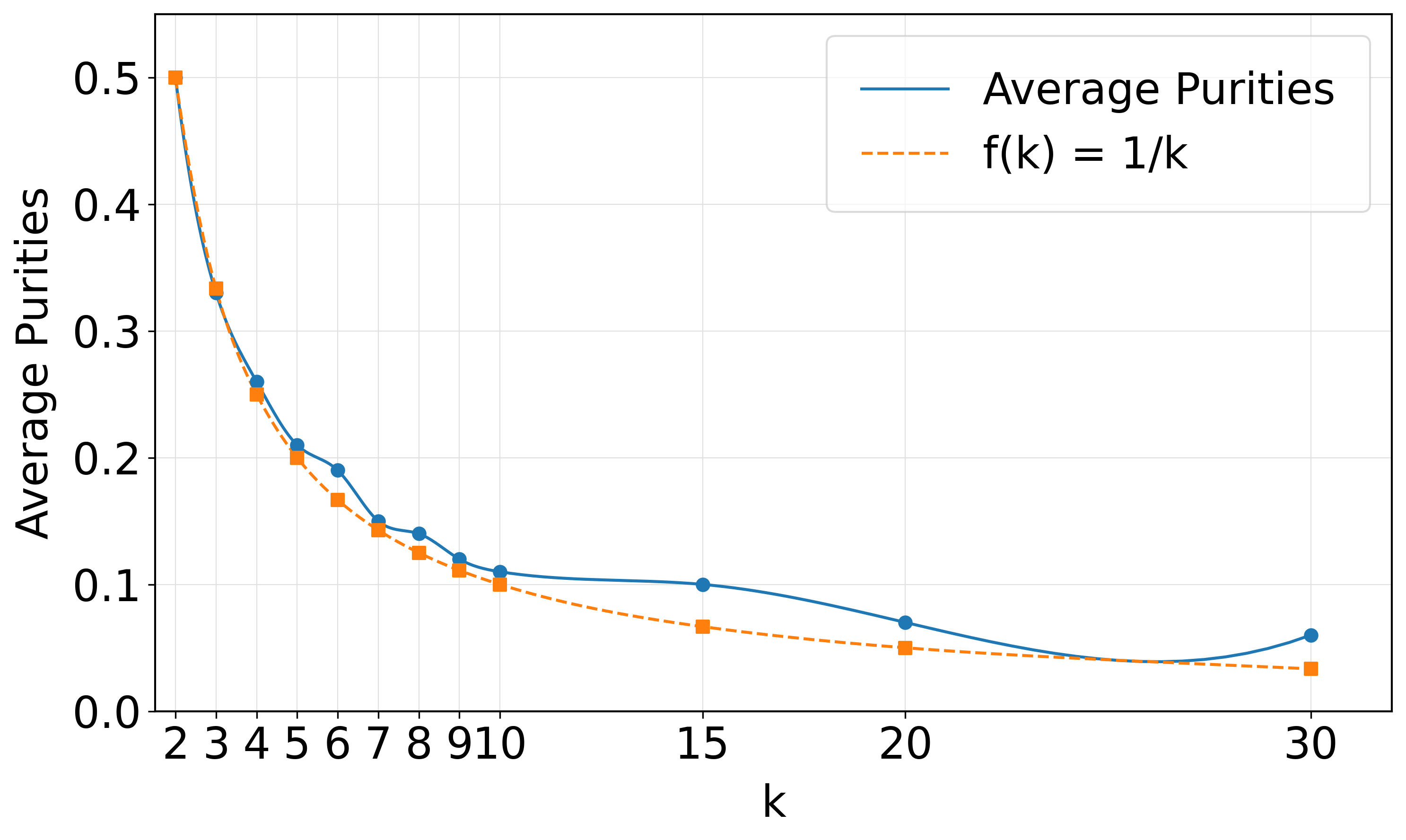}
    \caption{Average cluster purities vs.\ anonymity set size $k$, for $L=100$. The dashed line shows the theoretical lower bound $1/k$.}
    \label{fig:average_purities}
\end{figure}

% This result confirms that the Tamaraw-specific distance function naturally leads to anonymity sets composed of traffic patterns from multiple websites, supporting the desired $l$-diversity property of the defense.

\subsection{CAST Modifications}
\label{appendix:cast}

To assess the effect of our CAST algorithm changes, we performed an ablation study measuring each modification’s impact on clustering metrics using the Sirinam et al. dataset (Table \ref{tab:clustering-variants}). The baseline CAST creates highly fragmented clusters, with over 81\% containing fewer than 20 traces. Adding a dynamic threshold increases average cluster size, while cleaning and post-processing improve size balance. Cleaning reduces small clusters from 64.3\% to 13.6\%, and post-processing merges remaining ones to 0.2\%, yielding a final size ratio of 3.83 and an average cluster size of 200. Together, these updates produce stable, well-balanced traffic pattern clusters vital to our defense framework.

% To evaluate the impact of our CAST algorithm modifications, we conducted an ablation study measuring each change's contribution to clustering metrics across the Sirinam et al. dataset in Table \ref{tab:clustering-variants}. The baseline CAST yields excessively fragmented clusters, with over 81\% containing fewer than 20 traces and showing severe size imbalance. Introducing a dynamic threshold increases average cluster size, while subsequent cleaning and post-processing steps also improve size balance. Cleaning reduces small clusters from 64.3\% to 13.6\% and improves size ratios, and post-processing further merges residual small clusters (down to 0.2\%), achieving a final size ratio of 3.83 and average cluster size of 200. These incremental improvements collectively yield stable, uniformly sized traffic pattern clusters essential for our defense framework.
\begin{table}[htbp]
\centering
\caption{Effects of CAST modifications on clustering quality.}
\footnotesize
\label{tab:clustering-variants}
\small
\begin{tabular}{l|ccc}
\toprule
\textbf{Variant} & \textbf{Average} & \textbf{Small Cluster} & \textbf{Largest/} \\
 & \textbf{Cluster Size} & \textbf{\%} & \textbf{Smallest} \\
\midrule
Baseline (Raw) & 13 & 81.9 & 200 \\
+ Dynamic Threshold & 70 & 64.3 & 633 \\
+ Cleaning Step & 74 & 13.6 & 27.7 \\
+ Post Processing & 200 & 0.2 & 3.83 \\
\bottomrule
\end{tabular}
\end{table}

\subsection{Sensitivity Analysis on $\alpha$}
 The early-switching confidence threshold $\alpha$ controls the balance between efficiency and security in our adaptive defense by defining the safe timestamp $\tau_s$ where the system switches from global to local configuration. We derived anonymity sets from the training set and ran a sensitivity study on validation data. Lower $\alpha$ values trigger earlier, less confident switching, extending the use of efficient local parameters but increasing vulnerability, as predictions align less with true anomaly structures. As Table~\ref{tab:alpha-performance} shows, reducing $\alpha$ from 1.0 to 0.6 boosts overhead savings from $\sim$15\% to $\sim$34\%, while the $tau_s$ attacker accuracy (introduced in Appendix \ref{appendix:sec_bound}) rises from 34\% to 41\%. We chose $\alpha = 0.9$ for its balanced trade-off.
\begin{table}[htbp]
\centering
\caption{Performance metrics across different alpha values for adaptive defense mechanism.}
\label{tab:alpha-performance}
\footnotesize
\begin{tabular}{l|ccccc}
\toprule
\diagbox[width=10em,height=2.5\baselineskip]{Metrics}{$\alpha$} & \textbf{0.6} & \textbf{0.7} & \textbf{0.8} & \textbf{0.9} & \textbf{1.0} \\
\midrule
Overhead Improvement & -34\% & -32\% & -27\% & -25\% & -15\% \\
$\tau_s$ Aware Attacker Acc. & 41\% & 39\% & 38\% & 36\% & 34\% \\
\bottomrule
\end{tabular}
\end{table}

\subsection{Webpage vs Pattern Analysis}

 To quantify the end-to-end advantage of clustering fine-grained traffic patterns rather than the conventional approach of clustering at the website level, we conducted an ablation study comparing the performance of Adaptive Tamaraw under both configurations, using a bucket length of $L = 100$ on the Sirinam et al. dataset. Table~\ref{tab:granularity-results} presents the resulting average total overhead (bandwidth + time), confirming that operating at the pattern level consistently yields lower overhead across all tested values of $k$.  For instance, at $k = 2$, our fine-grained approach reduces overhead by 11 percentage points compared to the website-level baseline. This performance gain stems from the increased homogeneity of pattern-level anonymity sets.
\begin{table}[htbp]
\centering
\caption{Combined Overheads across different granularity levels.}
\label{tab:granularity-results}
\small
\begin{tabular}{l|ccc}
\toprule
\textbf{Granularity} & \textbf{$k=2$} & \textbf{$k=7$} & \textbf{$k=30$} \\
\midrule
Website Level & 219\% & 229\% & 243\% \\
Pattern Level & \textbf{205}\% & \textbf{217}\% & \textbf{235}\% \\
\bottomrule
\end{tabular}
\end{table}

\section{Computational and Memory Overhead}
\label{appendix:costs}

We report the computational and memory costs of Adaptive Tamaraw, measured on a system with an AMD EPYC 9454 CPU and an NVIDIA H100 GPU (20 GB MIG). The main Holmes model is 8.21 MB, and each lightweight kFP model is 0.41 MB. Inference latency per decision point is 1.84 ms: 0.15 ms for site prediction and 1.69 ms for traffic pattern identification. Our ECDIRE procedure retains 4.09 safe timestamps per site, requiring about 160 MB total storage for 95 sites. As summarized in Table~\ref{tab:comp-costs}, the defense is lightweight and suitable for real-time deployment.

% We report the computational and memory costs of Adaptive Tamaraw. All performance metrics were measured on a system equipped with an AMD EPYC 9454 48-Core CPU and an NVIDIA H100 GPU (20GB MIG instance). The primary Holmes model is compact, with a size of 8.21 MB, while each lightweight kFP model is only 0.41 MB. The total inference latency per decision point is low at 1.84 ms, which consists of 0.15 ms for the Holmes model to predict the webpage and 1.69 ms for the corresponding kFP model to identify the traffic pattern. On average, our ECDIRE procedure retains 4.09 safe timestamps per website. For a dataset like Sirinam et al. with 95 sites, the total storage requirement for all necessary kFP models is therefore approximately 160 MB (95 sites $\times$ 4.09 timestamps/site $\times$ 0.41 MB/model). These metrics, summarized in Table~\ref{tab:comp-costs}, demonstrate that our defense is lightweight and suitable for real-time deployment.
\begin{table}[h]
\centering
\caption{Summary of computational and resource costs.}
\label{tab:comp-costs}
\small
\begin{tabular}{l|c}
\hline
\textbf{Metric} & \textbf{Value} \\
\hline
Holmes Model Size & 8.21 MB \\
kFP Model Size (per instance) & 0.41 MB \\
Total Inference Latency (per decision) & 1.84 ms \\
Avg. Safe Timestamps per Site & 4.09 \\
\hline
\end{tabular}
\end{table}

\section{Proof of Security Bound for Adaptive Tamaraw}
\label{appendix:proof}

% \textbf{Motivation.}\;%
% Classical Tamaraw provides information-theoretic protection, but the proof holds only
% for one fixed (and costly) send rate.  
% \textit{Adaptive Tamaraw} keeps a formal success bound while switching to
% lighter, AS-specific rates after a safe time~\(t_{S_i}\).
% We prove the resulting two-phase defence is still
% \(\delta\)-non-injective (as described in Section~\ref{sec:problem_statement}); numerical values of the bound
% appear later in the evaluation section.

\textbf{Notation.}\;
Let \(\mathcal{S}= \{S_1,\dots,S_m\}\) be the anonymity sets of traffic patterns (produced in
Section~\ref{sec:AS_gen}).
For every set \(S_i\) we fix a \emph{safe time} \(t_{S_i}\); this value is the same for
\emph{all} traces in \(S_i\) but may differ across sets.
  
% We use two padding operators:

% \smallskip
% \begin{itemize}
% \item
% \(\displaystyle
%    D_{\text{glob}}(\,\cdot\,;\,p_{\mathrm{in}}^{g},p_{\mathrm{out}}^{g})
% \)  
% applies the \emph{global} constant–rate schedule with parameters
% \(p_{\mathrm{in}}^{g},p_{\mathrm{out}}^{g}\).

% \item
% \(\displaystyle
%    D_{S_i}(\,\cdot\,;\,p_{\mathrm{in}}^{(S_i)},p_{\mathrm{out}}^{(S_i)},L)
% \)  
% applies the \emph{AS–specific} schedule chosen for \(S_i\).
% \end{itemize}

% Given a raw trace \(f\in S_i\) the full defense acts in two stages:
% \[
%   \widehat{f}=D(f)
%   =
%   \begin{cases}
%        D_{\text{glob}}\!\bigl(f(t);\;p_{\mathrm{in}}^{g},p_{\mathrm{out}}^{g}\bigr),
%        & 0 \le t < t_{S_i},\\[6pt]
%        D_{S_i}\!\bigl(f(t);\;p_{\mathrm{in}}^{(S_i)},p_{\mathrm{out}}^{(S_i)},L\bigr),
%        & t \ge t_{S_i}.
%   \end{cases}
% \]
% We write this concisely as a labelled transition
% \[
%    D_{\text{glob}}
%    \;\xrightarrow[\;t_{S_i}\;]{}
%    D_{S_i}.
% \]
Because \(t_{S_i}\) is fixed \emph{inside} the set, an observer who notices
the switch time learns only that the trace belongs to an anonymity set
whose safe time equals~\(t_{S_i}\).

\medskip
\textbf{AS-specific tail.} Inside an anonymity set \( S_i \), let
\[
S_{i,\ell} = \left\{ t \in S_i : \operatorname{len}(t) = \ell \right\}
\]
denote the subset of defended traces whose total length is \( \ell \).
Each \( S_{i,\ell} \) corresponds to a bucket of traces that become indistinguishable after regularization. For each such bucket, an optimal attacker performs a majority vote, achieving success rate
\[
\frac{\max_w \left| \{ t \in S_{i,\ell} : \operatorname{site}(t) = w \} \right|}{|S_{i,\ell}|}.
\]

To model the average-case attacker success over \( S_i \), we define the \emph{non-uniform attacker accuracy} as the expected success rate across all observed length buckets. Let \(\mathcal{L}\) denote the set of all distinct length buckets observed in \(S_i\). Then:

\begin{align}
\label{eq:A_bar}
\bar{A}(S_i) &= \sum_{\ell \in \mathcal{L}} \frac{|S_{i,\ell}|}{|S_i|} \cdot 
\frac{\max\limits_{w} \left| \left\{ t \in S_{i,\ell} : \operatorname{site}(t) = w \right\} \right|}{|S_{i,\ell}|} \\
&= \mathbb{E}_{\ell \sim \mathcal{L}} \left[
\frac{\max\limits_{w} \left| \left\{ t \in S_{i,\ell} : \operatorname{site}(t) = w \right\} \right|}{|S_{i,\ell}|}
\right]
\end{align}

% \textbf{AS-specific tail.}\;%
% Inside an anonymity set \(S_i\) let
% \[
%    S_{i,\ell} = \{\,t\in S_i : \text{len}(t)=\ell\}
% \]
% be the bucket of defended traces whose total length is \(\ell\).
% With a majority vote the optimal attacker achieves
% \begin{equation}
%    A_{\max}(S_i)
%    \;=\;
%    \max_{\ell}
%      \frac{\displaystyle
%        \max_{w}|\,\{t\in S_{i,\ell}: \text{site}(t)=w\}\,|}
%           {|S_{i,\ell}|}.
%    \label{eq:Amax}
% \end{equation}

\begin{lemma}[Post-switch non-uniform weighted $\delta$] 
\label{lem:delta-i}
For every anonymity set $S_i$, the AS-specific mapping $D_{S_i}$ is
\emph{non-uniformly} $\delta_i$-weighted-non-injective, where
$\displaystyle \delta_i = {1}/{\bar{A}(S_i)}$ and
$\bar{A}(S_i)$ is the average attacker success rate defined in~\eqref{eq:A_bar}.
\end{lemma}

\begin{proof}
After the switch, all traces in $S_i$ are padded using the fixed parameters
$(p_{\mathrm{in}}^{(S_i)}, p_{\mathrm{out}}^{(S_i)}, L)$, so the only
observable feature is the total defended length~$\ell$.
This partitions $S_i$ into buckets
\[
S_{i,\ell} = \left\{ t \in S_i : \operatorname{len}(t) = \ell \right\}.
\]

\smallskip
\emph{Weighted pre-image size.}
Each bucket length $\ell$  (i.e., each trace \(t \in S_{i,\ell}\)) corresponds to a single defended trace
\(f'_{i,\ell}\) . The weighted
pre-image size of \(f'_{i,\ell}\), by~\eqref{eq:weighted-delta}, is
\[
  \tilde{\delta}\left(f'_{i,\ell}\right)
  = \frac{|S_{i,\ell}|}
         {\displaystyle\max_{w} \left| \left\{ t \in S_{i,\ell} : \text{site}(t) = w \right\} \right| }.
\]

\smallskip
\emph{Attacker success.}
The optimal attacker outputs the majority website in each bucket, achieving
success \(1/\tilde{\delta}(f'_{i,\ell})\).

\smallskip
\emph{Average-case accuracy.}
The attacker’s \emph{expected} success rate over $S_i$ is given by the non-uniform accuracy
\[
\bar{A}(S_i)
= \sum_{\ell} \frac{|S_{i,\ell}|}{|S_i|} \cdot \frac{1}{\tilde{\delta}(f'_{i,\ell})} =  \mathbb{E}_{f' \sim S_i} \left[\frac{1}{\tilde{\delta}(f')} \right],
\]
\smallskip
Therefore, by the definition in Section~\ref{sec:problem_statement}, the mapping $D_{S_i}$ is non-uniformly $\delta_i$-weighted-non-injective with $\delta_i = 1/\bar{A}(S_i)$. It follows that the attacker's overall success rate for any trace in $S_i$ is bounded by
\[
\Pr[\text{success} \mid S_i] \;\le\; \bar{A}(S_i). \qedhere
\]
\end{proof}

% \medskip
% %-----------------------------------------------------------
% % Two–phase success bound
% %-----------------------------------------------------------
% \textbf{Global prefix (per-set).}\;\;
% Because every trace is initially padded using the global parameters, the attacker observes a uniform prefix up to the set-specific switching time \(\tau_{S_i}\). At this point, all defended traces in the same anonymity set \(S_i\) are indistinguishable, and the attacker’s best strategy is to guess the most frequent website in \(S_i\). We define

% \[
% \varepsilon_{g_i}
% \;=\;
% \max_{w\in\mathcal{W}}
% \Pr\left[\text{site}(t)=w\midt\inS_i\right],
% \]

% i.e., the attacker’s optimal success rate using only the global prefix before \(\tau_{S_i}\). This value reflects the lack of \(l\)-diversity in the set: if one website dominates \(S_i\), then \(\varepsilon_{g_i}\) will be high. In contrast, if \(S_i\) contains traces from many distinct websites, \(\varepsilon_{g_i}\)approaches uniform guessing.

% \textbf{Two–phase bound.}\;\;
% For a connection whose defended tail belongs to anonymity set \(S_i\), the attacker’s overall success is bounded by

% \[
% \Pr[\text{success}\midS_i]\;\le\;\max\big\{\varepsilon_{g_i},\,\bar{A}(S_i)\big\},
% \]

% where \(\varepsilon_{g_i}\) captures the attacker’s best chance using only the global prefix, and \(\bar{A}(S_i)\) is the non-uniform attacker accuracy when classifying within \(S_i\)’s padded tail.

\begin{theorem}[Global non-uniformly weighted $\delta$–non-injectivity]
\label{thm:delta-global}
Let $\mathcal{S}$ be the set of all anonymity sets obtained, and for each anonymity set $S_i \in \mathcal{S}$, let
\[
\delta_i = \frac{1}{\bar{A}(S_i)}.
\]
Define the global non-uniform weighted injectivity parameter \(\delta\) as
\[
\frac{1}{\delta} = \mathbb{E}_{S_i \sim \mathcal{S}} \left[\frac{1}{\delta_i} \right] = \mathbb{E}_{S_i}[\bar{A}(S_i)].
\]
Then, \textit{Adaptive Tamaraw} is non-uniformly weighted $\delta$-non-injective, and the attacker’s average success probability is bounded by
\[
\boxed{\Pr[\text{success}] \;\le\; \frac{1}{\delta}}.
\]
\end{theorem}

\begin{proof}
Recall from Lemma~\ref{lem:delta-i} that for each anonymity set \(S_i\), the attacker’s success rate over its defended outputs is bounded by \(\bar{A}(S_i) = \mathbb{E}_{f' \sim S_i}\left[ \frac{1}{\tilde{\delta}(f')} \right]\), and thus the mapping \(D_{S_i}\) is non-uniformly \(\delta_i\)-weighted-non-injective, where \(\delta_i = 1/\bar{A}(S_i)\).

Now consider the overall defended distribution \(\mathcal{F}'\), which consists of all post-regularization traces \(f'_{i,\ell}\) produced from inputs in various anonymity sets \(S_i\). Each defended trace \(f'\) must originate from some \(S_i\) with probability \(P(S_i)\), and within \(S_i\), from a bucket of length \(\ell\) with probability \(|S_{i,\ell}|/|S_i|\).

We begin with the average inverse diversity over anonymity sets:
\begin{align*}
\mathbb{E}_{S_i}\left[\frac{1}{\delta_i}\right]
&= \mathbb{E}_{S_i}\left[\bar{A}(S_i)\right] \\
&= \sum_i P(S_i) \cdot \bar{A}(S_i) \\
&= \sum_i P(S_i) \cdot \sum_{\ell} \frac{|S_{i,\ell}|}{|S_i|} \cdot \frac{1}{\tilde{\delta}(f'_{i,\ell})}.
\end{align*}

Note that \(P(S_i) \cdot \frac{|S_{i,\ell}|}{|S_i|}\) is precisely the probability of observing the defended trace \(f'_{i,\ell}\) in the output space \(\mathcal{F}'\), since:
\[
P(f'_{i,\ell}) = P(S_i) \cdot \frac{|S_{i,\ell}|}{|S_i|}.
\]

Also, \(\tilde{\delta}(f'_{i,\ell})\) is the weighted pre-image size of \(f'_{i,\ell}\) as defined in Section~\ref{sec:problem_statement}.

We can therefore rewrite the sum as:
\begin{align*}
\sum_i \sum_{\ell} P(S_i) \cdot \frac{|S_{i,\ell}|}{|S_i|} \cdot \frac{1}{\tilde{\delta}(f'_{i,\ell})}
&= \sum_{f' \in \mathcal{F}'} P(f') \cdot \frac{1}{\tilde{\delta}(f')} \\
&= \mathbb{E}_{f' \sim \mathcal{F}'}\left[ \frac{1}{\tilde{\delta}(f')} \right].
\end{align*}

Thus, we have:
\[
\mathbb{E}_{S_i}\left[\frac{1}{\delta_i}\right] = \mathbb{E}_{f' \sim \mathcal{F}'}\left[ \frac{1}{\tilde{\delta}(f')} \right],
\]
which confirms that the expected attacker success over anonymity sets equals the expected attacker success over the defended output distribution.

Therefore, \emph{Adaptive Tamaraw} as a whole satisfies the non-uniform \(\delta\)-non-injectivity condition with
\[
\delta = \left( \mathbb{E}_{f' \sim \mathcal{F}'}\left[ \frac{1}{\tilde{\delta}(f')} \right] \right)^{-1},
\]
completing the proof.
\end{proof}